\documentclass{lmcs}
\pdfoutput=1
\usepackage[utf8]{inputenc}

\usepackage{lastpage}
\lmcsdoi{22}{2}{4}
\lmcsheading{}{\pageref{LastPage}}{}{}%
{Mar.~21,~2024}{Apr.~15,~2026}{}

\usepackage{algorithmic}
\usepackage[ruled]{algorithm2e}
\usepackage{tikz}
\usepackage{color}
\usepackage{adjustbox}
\usepackage{pgfplots}
\usepackage{hyperref}
\usepackage{cleveref}
\usepackage{amsmath}
\usepackage{subcaption}
\usepackage{placeins}
\usepackage{booktabs}
\usepackage{graphicx}

\makeatletter         
\def\figurecaption#1#2{\noindent\hangindent 40pt
                       \hbox to 36pt {\small\sl #1 \hfil}
                       \ignorespaces {\small #2}}

\long\def\@makecaption#1#2{
  \vskip 10pt 
  \settowidth{\@tempdima}{#2}
  \ifdim\@tempdima>0pt
       \setbox\@tempboxa\hbox{#1: #2}
     \else
       \setbox\@tempboxa\hbox{#1 #2}
   \fi
   \ifdim \wd\@tempboxa >\hsize               
       \begin{list}{#1:}{
       \settowidth{\labelwidth}{#1:}
       \setlength{\leftmargin}{\labelwidth}
       \addtolength{\leftmargin}{\labelsep}
        }\item #2 \end{list}\par   
     \else                                    
       \hbox to\hsize{\hfil\box\@tempboxa\hfil}  
   \fi}
\makeatother

\def\Real{\mathbb{R}}
\def\Nat{\mathbb{N}}

\def\probeq{\bm{\sim}}
\DeclareMathOperator{\relu}{ReLU}
\DeclareMathOperator{\Reach}{Reach}
\DeclareMathOperator{\ReLU}{\relu}
\DeclareMathOperator{\Vars}{Vars}
\DeclareMathOperator{\softplus}{Softplus}
\DeclareMathOperator{\Leb}{Leb}

\usepackage{listings}
\usepackage{bm}

\newcommand{\keyword}[1]{\bm{\mathtt{#1}}}

\DeclareUnicodeCharacter{2212}{−} 
\usepgfplotslibrary{groupplots,dateplot}
\usetikzlibrary{patterns,shapes.arrows}
\pgfplotsset{compat=newest}

\newcommand{\indicating}{\text{ind}}
\newcommand{\nonincreasing}{\text{non-inc}}
\newcommand{\bound}{\text{min}}
\newcommand{\verifbound}{\text{bound}}

\crefname{lstlisting}{listing}{listings}
\Crefname{lstlisting}{Listing}{Listings}

\tikzset{
dotted ends/.style={
draw=none,
postaction={
   draw, decoration={lineto, 
   pre=moveto, pre length=0.05*\pgfmetadecoratedpathlength,
   post=moveto, post length=0.05*\pgfmetadecoratedpathlength, 
   }, decorate},
postaction={
   draw, dotted, decoration={lineto, 
   post=moveto, post length=0.95*\pgfmetadecoratedpathlength, 
   }, decorate},
postaction={
   draw, dotted, decoration={lineto, 
   pre=moveto, pre length=0.95*\pgfmetadecoratedpathlength,
   }, decorate},
}}

\pgfplotstableread{
-1 0.0
0 0.0
1 0.00431
2 0.00892
3 0.01678
4 0.02967
5 0.0518
6 0.08709
7 0.14061
8 0.23266
9 0.37979
10 0.61491
11 1.0 
11 1.0
}\repulsedata

\pgfplotstableread{
-3 0.0
-1 0.0
0 0.0
0.01 0.00431
1 0.00431
1.01 0.00892
2 0.00892
2.01 0.01678
3 0.01678
3.01 0.02967
4 0.02967
4.01 0.0518
5 0.0518
5.01 0.08709
6 0.08709
6.01 0.14061
7 0.14061
7.01 0.23266
8 0.23266
8.01 0.37979
9 0.37979
9.01 0.61491
10 0.61491
10.01 1.0
11 1.0 
12 1.0
}\repulsedatastepfunc

\pgfplotstableread{
-1 0.0
0 0.0
0.5 0.00431
1.5 0.00892
2.5 0.01678
3.5 0.02967
4.5 0.0518
5.5 0.08709
6.5 0.14061
7.5 0.23266
8.5 0.37979
9.5 0.61491
10.5 1.0
}\repulsedatamod

\pgfplotstableread{
-1 0.0
0 0.0
1 0.00431
2 0.00892
3 0.01678
4 0.02967
5 0.0518
6 0.08709
7 0.14061
8 0.23266
9 0.37979
10 0.61491
11.5 1.192545
} \certdata

\pgfplotstableread{
0.0 0.0
0.1 0.00425
0.2 0.00403
0.30000000000000004 0.00467
0.4 0.00418
0.5 0.00453
0.6000000000000001 0.00427
0.7000000000000001 0.00425
0.8 0.00419
0.9 0.00453
1.0 0.00439
1.1 0.00879
1.2000000000000002 0.0084
1.3 0.00835
1.4000000000000001 0.00828
1.5 0.00842
1.6 0.00889
1.7000000000000002 0.00803
1.8 0.00836
1.9000000000000001 0.0086
2.0 0.00857
2.1 0.01781
2.2 0.0174
2.3000000000000003 0.01698
2.4000000000000004 0.01729
2.5 0.01773
2.6 0.0167
2.7 0.01696
2.8000000000000003 0.01778
2.9000000000000004 0.01733
3.0 0.01757
3.1 0.03031
3.2 0.03022
3.3000000000000003 0.03072
3.4000000000000004 0.02998
3.5 0.03031
3.6 0.0295
3.7 0.03071
3.8000000000000003 0.03108
3.9000000000000004 0.03068
4.0 0.0299
4.1000000000000005 0.05232
4.2 0.05238
4.3 0.05117
4.4 0.0521
4.5 0.05124
4.6000000000000005 0.05185
4.7 0.05156
4.800000000000001 0.04973
4.9 0.05051
5.0 0.05207
5.1000000000000005 0.08459
5.2 0.08454
5.300000000000001 0.08649
5.4 0.08514
5.5 0.08561
5.6000000000000005 0.08621
5.7 0.0849
5.800000000000001 0.08583
5.9 0.08537
6.0 0.08728
6.1000000000000005 0.14053
6.2 0.14101
6.300000000000001 0.14123
6.4 0.14262
6.5 0.14295
6.6000000000000005 0.14263
6.7 0.14034
6.800000000000001 0.14141
6.9 0.14208
7.0 0.1443
7.1000000000000005 0.23368
7.2 0.23234
7.300000000000001 0.23228
7.4 0.23371
7.5 0.23246
7.6000000000000005 0.23276
7.7 0.2313
7.800000000000001 0.23215
7.9 0.23254
8.0 0.23239
8.1 0.37999
8.200000000000001 0.38014
8.3 0.37852
8.4 0.37978
8.5 0.37976
8.6 0.37825
8.700000000000001 0.38067
8.8 0.3803
8.9 0.37911
9.0 0.38223
9.1 0.61614
9.200000000000001 0.61709
9.3 0.6156
9.4 0.61455
9.5 0.61693
9.600000000000001 0.61679
9.700000000000001 0.61833
9.8 0.61673
9.9 0.61456
10.0 0.61644
10.100000000000001 1.0
10.200000000000001 1.0
10.3 1.0
10.4 1.0
10.5 1.0
10.600000000000001 1.0
10.700000000000001 1.0
10.8 1.0
10.9 1.0
11.0 1.0
}\repulsedatafiner

\theoremstyle{plain}\newtheorem{lemma}[thm]{Lemma} 
\theoremstyle{plain}\newtheorem{theorem}[thm]{Theorem}

\title{Quantitative Verification with Neural Networks}
\titlecomment{{\lsuper*}{The conference version of this manuscript appeared at CONCUR 2023~\cite{DBLP:conf/concur/AbateEGPR23}}}

\thanks{This work was supported in part by 
Innovate UK under the HICLASS project, 
EPSRC Centre for Doctoral Training in Autonomous Intelligent Machines and Systems (EP/S024050/1) and EPSRC Doctoral Training Partnership,
Google DeepMind Scholarship,
Advanded Research and Invention Agency under the Safeguarded AI programme.
}

\author[A.~Abate]{Alessandro Abate\lmcsorcid{0000-0002-5627-9093}}[a]
\author[A.~Edwards]{Alec Edwards\lmcsorcid{0000-0001-9174-9962}}
\author[M.~Giacobbe]{Mirco Giacobbe\lmcsorcid{0000-0002-1825-0097}}[b]
\author[H.~Punchihewa]{Hashan Punchihewa}
\author[D.~Roy]{Diptarko Roy\lmcsorcid{0009-0003-4306-2076}}[b]
\address{Department of Computer Science, University of Oxford, UK}	
\email{alessandro.abate@cs.ox.ac.uk}  

\address{School of Computer Science, University of Birmingham, UK}	
\email{m.giacobbe@bham.ac.uk, d.s.roy@bham.ac.uk}  


\begin{document}

\begin{abstract} 
  We present a data-driven approach to the quantitative verification of probabilistic programs and stochastic dynamical models. Our approach leverages neural networks to compute tight and sound bounds for the probability that a stochastic process hits a target condition within finite time. This problem subsumes a variety of quantitative verification questions, 
  from the reachability and safety analysis of discrete-time stochastic dynamical models, 
  to the study of assertion-violation and termination analysis of probabilistic programs. 
  We rely on neural networks to represent supermartingale certificates that yield such probability bounds,  
  which we compute using a counterexample-guided inductive synthesis 
  loop: 
  we train the neural certificate while tightening the probability bound over samples of the state space using stochastic optimisation, and then we formally check the certificate's validity over every possible state using satisfiability modulo theories; if we receive a counterexample, we add it to our set of samples and repeat the loop until validity is confirmed.  
  We demonstrate on a diverse set of benchmarks 
  that, thanks to the expressive power of neural networks, 
  our method yields smaller or comparable probability bounds than existing symbolic methods in all cases, 
  and that our approach succeeds on models that are entirely beyond the reach of such alternative techniques. 
\end{abstract}

\maketitle

\section{Introduction}

Probabilistic programs extend imperative programs with the ability to sample from probability distributions~\cite{semantics,GordonHNR14,Kozen81,McIverM05}, 
which provides an expressive language to describe randomized algorithms, cryptographic protocols, 
and Bayesian inference schemes. 
Discrete-time stochastic dynamical models, 
characterised by stochastic difference equations, 
are a natural framework to describe auto-regressive time series, as well as sequential decision and planning problems in unknown environments. 
A fundamental quantitative verification problem for probabilistic programs and stochastic dynamical models is the quantitative reachability question, 
which amounts to finding the probability with which the system reaches a given target condition within a finite number of steps. 
Reachability is at the core of a variety of other important quantitative verification questions as, 
by selecting appropriate target conditions on the state space, we can express the probability that a probabilistic program 
terminates or that it violates an assertion, as well as the probability that a stochastic dynamical model satisfies an invariant or remains within a set of safe configurations. 

Quantitative reachability verification has been studied extensively using theories and algorithms built upon symbolic reasoning techniques,  
such as quantitative calculi~\cite{DBLP:conf/lpar/McIverM02,DBLP:journals/toplas/MorganMS96}, 
probabilistic model checking algorithms~\cite{DBLP:journals/entcs/KwiatkowskaNP06,DBLP:conf/tacas/ForejtKNPQ11}, 
discrete abstractions of stochastic dynamical models~\cite{DBLP:conf/hybrid/AbateKM11,DBLP:journals/iandc/TkachevMKA17,SA13}, 
and the synthesis of supermartingale-like certificates~\cite{DBLP:conf/cav/ChakarovS13,DBLP:conf/cav/ChatterjeeFG16,DBLP:conf/popl/ChatterjeeNZ17,DBLP:journals/toplas/ChatterjeeFNH18,LSAZ22}. 
Among the latter class, a method to provide a sound upper bound for the reachability probability of a system 
is to synthesise a supermartingale function that maps every reachable state to a non-negative real, 
whose value is never smaller than 1 inside the target condition, 
and such that it never increases in expectation as the system evolves outside of the target.  
This is referred to as a~\textit{non-negative repulsing supermartingale} or \textit{stochastic invariant indicator} in the literature~\cite{DBLP:journals/toplas/TakisakaOUH21,DBLP:conf/cav/ChatterjeeGMZ22,DBLP:conf/popl/ChatterjeeNZ17}, 
and its output over a given state provides an upper bound for the probability that the system reaches the target condition from that state.
Symbolic methods for the synthesis of such certificates assume that the supermartingale function, 
as well as its post-expectation, which depends on the model constraints and distributions, 
are both in linear or polynomial form. This poses syntactic restrictions to their applicability. 

Data-driven and counterexample-guided inductive synthesis (CEGIS) procedures, 
combined with machine learning approaches that leverage neural networks to represent certificates, 
have shown great promise in mitigating the aforementioned limitation~\cite{DBLP:conf/nips/ChangRG19,DBLP:journals/csysl/AbateAGP21,giacobbe2022,DBLP:conf/tacas/PeruffoAA21,DBLP:journals/csysl/MathiesenCL23,DBLP:conf/cav/AbateGR20,DBLP:conf/aaai/LechnerZCH22,DBLP:conf/aaai/ZikelicLHC23,DBLP:conf/nips/ZikelicLVCH23,neuralmc,DBLP:conf/aaai/NadaliM0024,DBLP:conf/aaai/NeustroevGL25}. 
In particular, neural-based CEGIS decouples the task of guessing a certificate from that of checking its validity, 
delegating the guessing task to efficient machine learning algorithms that leverage the expressive power of neural networks, 
while confining symbolic reasoning to the checking part of the task, 
which is computationally easier to solve in isolation than the entire synthesis problem. 
CEGIS has been applied to the synthesis of {\em neural supermartingales} for the almost-sure termination analysis of 
probabilistic programs~\cite{DBLP:conf/cav/AbateGR20}, 
as well as their counterpart for stochastic dynamical models with applications to qualitative queries such as almost-sure 
safety and stability~\cite{DBLP:journals/csysl/MathiesenCL23,DBLP:conf/aaai/LechnerZCH22}.

In this paper, we present theory, methods, and an extensive experimental evaluation to demonstrate the efficacy and flexibility of neural supermartingales to solve {\em quantitative verification} questions for probabilistic programs and stochastic dynamical models, using machine learning combined with satisfiability modulo theories (SMT) technologies.

\begin{description}
    \item[Theory] We use the theory of non-negative repulsing supermartingales \cite[Section 4]{DBLP:journals/toplas/TakisakaOUH21} to leverage neural networks 
    as representations of supermartingale functions for quantitative verification, while exploiting SMT solving to guarantee soundness. 
    
    \item[Methods] We present a CEGIS-based approach to train neural supermartingale functions that minimise an upper bound for the reachability  probability over sample points from the state space, and check their validity over every possible state using SMT solving. 
     We present a program-agnostic approach that relies on state samples in the training phase, and also a novel 
     program-aware approach that embeds model information in the loss function to enhance the effectiveness of stochastic optimisation.
    \item[Experiments] We build a prototype implementation and compare the efficacy of our method with the
    state of the art in synthesis of linear supermartingales using symbolic reasoning~\cite{DBLP:journals/toplas/TakisakaOUH21}. 
    We show that our program-aware approach computes tighter or comparable 
    probability bounds than symbolic reasoning on existing benchmarks, while our program-agnostic approach matches it in over half of the instances. 
    We additionally demonstrate that both our approaches can handle models beyond reach of purely symbolic methods.
\end{description}

\section{Probabilistic Programs and Stochastic Dynamical Models}

\begin{figure}[t]
  \begin{align*}
    v &\in \mbox{Vars} & \text{(variables)}\\
    N &\in \Real & \text{(numerals)}
    \\E &::= v \mid N \mid \texttt{-}E \mid E+E \mid E-E \mid E*E \mid \dots
    &\text{(arithmetic expressions)}
    \\P &::= {\tt Bernoulli}\texttt{(}\,E\,\texttt{)} \mid
    {\tt Gaussian}\texttt{(}\,E\texttt{, }E\,\texttt{)} \mid \dots
    &\text{(probability distributions)}
    \\B &::= {\tt true} \mid {\tt !}B \mid B \,\texttt{\&\&}\, B \mid B\texttt{||}B  \mid E\,\texttt{==}\,E \mid E\,\texttt{<}\,E \mid \dots & \text{(Boolean expressions)}
    \\C &::= \keyword{skip} &\text{(update commands)}
    \\&\quad\mid v\,\texttt{=}\,E&\text{(deterministic assignment)}
    \\&\quad\mid v\,\probeq\,P&\text{(probabilistic assignment)}
    \\&\quad\mid C~{\tt ;}~C &\text{(sequential composition)}
    \\&\quad\mid \keyword{if}~B~\keyword{then}~C~\keyword{else}~C~\keyword{fi}&\text{(conditional composition)}
  \end{align*}
  \caption{Grammar for update commands, Boolean and arithmetic expressions.}
  \label{fig:syntax}
\end{figure}

Probabilistic programs are computer programs whose execution is determined by random variables, 
and stochastic dynamical models describe discrete-time dynamical systems with probabilistic behaviour. 
The earlier enjoy the flexibility of imperative programming constructs and are used to describe randomised algorithms, 
and the latter are expressed as stochastic difference equations and are used to describe probabilistic systems that evolve 
over infinite time. The semantics of both can be described in terms of  stochastic processes and, 
for this reason, verification questions for both  can be solved with similar techniques. 

The syntax of our modeling framework uses imperative constructs from probabilistic programs and defines 
executions over infinite time as dynamical systems. 
Specifically, we consider programs that operate over an ordered set of $n$ real-valued variables, denoted by $\Vars$, and update their values through the repeated execution of a command $C$ whose grammar is described in Figure~\ref{fig:syntax}. 
Under this definition, a state of the 
system is an $n$-dimensional vector $s \in \Real^n$ that assigns a value to each variable symbol. The update command $C$ 
defines an update function $f \colon \Real^n \times [0,1]^m \to \Real^n$, where $m$ is the number of syntactic probabilistic assignment statements occurring in $C$, which maps the current state and $m$ random variables uniformly distributed in $[0, 1]$ into the next state. 
Conceptually, within $f$, each random variable is mapped into its respective distribution by applying the appropriate inverse transformation. 
Altogether, our probabilistic program defines a stochastic process, whose 
behaviour is determined by the following stochastic difference equation: 
\begin{equation}
    s_{t+1} = f(s_t, r_t), \quad r_t \sim \mathbb{U}^m, 
    \label{eq:sys}
\end{equation}  
where $r_t$ is an $m$-dimensional random input sampled at time $t$ 
from the uniform distribution $\mathbb{U}^m$ over the $m$-dimensional hypercube $[0,1]^m$.  
The initial state $s_0$ is either given as a deterministic assignment 
to constant values, or is non-deterministically chosen from a set of initial 
conditions $S_0$ characterized by a Boolean expression.    
This setting can be seen as an assertion about the initial conditions followed by the probabilistic program
$\keyword{while~true~do}~C~\keyword{od}$.
As we show in \cref{sec:reachability}, this allows us to characterise verification 
questions such as termination, non-termination, invariance and 
assertion-violation for while loops with general guards, 
as well as reachability and safety verification questions for dynamical models with general stochastic disturbances.

The semantics of our model is defined as a stochastic process induced by the Markov chain over the probability space of infinite words of random samples.
This is defined by the probability space triple $(\Omega, \mathcal{F}, \mathbb{P})$ \cite{BS96,hll1996}, where 
\begin{itemize}
    \item $\Omega$ is the set of infinite sequences $\left([0, 1]^m\right)^\omega$ of $m$-dimensional tuples of values in $[0,1]$,  
    \item $\mathcal{F}$ is the extension of the Borel $\sigma$-algebra over the unit interval $\mathcal B([0, 1])$ to $\Omega$, 
    \item $\mathbb{P}$ is the extension of the Lebesgue measure on $[0, 1]$ to $\Omega$. 
\end{itemize}

Every initialisation of the system on state $s \in \Real^n$ induces a stochastic 
process $\{ X_t^{s} \}_{t \in \Nat}$ over the state space $\Real^n$. 
Let $\omega = r_0r_1r_2\dots$ be an infinite sequence of random samples in $[0,1]^m$, 
then the stochastic process is defined by the sequence of random variables 
\begin{equation}\label{eqn:stochprocess}
    X_{t + 1}^{s}(\omega) = 
    f(X_{t}^{s}(\omega), r_t), 
    \qquad X_0^{s}(\omega) = s.
\end{equation} 
This defines the natural filtration $\{ \mathcal{F}_t \}_{t \in \Nat}$, 
which is the smallest filtration to which the stochastic process $X^s_t$ is adapted. 
In other words \cite{kallenberg}, this can be seen as another Markov chain with state space $\Real^n$ and transition kernel 
\begin{equation}
    T(s, S') = \Leb\left( \left\{r \in [0, 1]^m \mid f(s,r) \in S'\right\} \right), 
\end{equation} 
where $S'$ is a Borel measurable subset of $\Real^n$ 
and Leb refers to the Lebesgue measure of a measurable subset of $[0, 1]^m$. In other words, kernel $T$ denotes the probability to transition from state $s$ into a set of states $S'$.  The transition kernel also defines 
the {\em post-expectation operator} $\mathbb X$, also known as the next-time operator \cite[Definition 2.16] {DBLP:journals/toplas/TakisakaOUH21}. 
$\mathbb X$ can be applied to an arbitrary Borel-measurable function $h \colon \Real^n \to \Real$ 
 defining the \textit{post-expectation of} $h$, denoted by $\mathbb{X}[h]$ and defined as the following function over states:
\begin{equation}\label{eqn:postex}
    \mathbb{X}[h](s)= \int 
    h(s')~T(s, \mathrm{d}s'). 
\end{equation}
This represents the expected value of $h$ evaluated at the next state, given the current state being $s$.
Computing the symbolic representation of a post-expectation for probabilistic programs and stochastic dynamical models 
is a core problem in probabilistic verification~\cite{gehr2016psi,gehr2020lpsi}. Indeed, our theoretical framework builds upon the post-expectation (cf.\ Eq.~\eqref{eq:spec-dec}), 
and our verification procedure uses a symbolic representation of the post-expectation (cf.~\cref{sec:verification}) as well as our program-aware learning procedure, while our program-agnostic learning procedure approximates it statistically (cf.~\cref{sec:learner}). 

We remark that our model encompasses general probabilistic program loops as well as stochastic dynamical systems 
with general disturbances. For example, a probabilistic loop with guard condition $B$ and body $C$ in the form
\begin{align}
    \keyword{while}~B~\keyword{do}~C~\keyword{od} \label{eq:loop}
\end{align}
can be expressed as a loop with guard $\keyword{true}$ and 
body $\keyword{if}~B~\keyword{then}~C~\keyword{else}~\keyword{skip}~\keyword{fi}$. 
This expresses the fact that, after termination, the program will stay on the terminal state indefinitely. 
Also, discrete-time stochastic difference equations with a nonlinear vector field $g$ and time-invariant input disturbance with an arbitrary distribution $\mathcal{W}$ in the form 
\begin{equation}
    s_{t+1} = g(s_t, w_t), \quad w_t \sim \mathcal{W}, \label{eq:dynsys}
\end{equation} 
comply with our model. It is sufficient to derive $w_t$ with an appropriate 
inverse transformation from the uniform distribution and embed it in the update function $f$ (Eq.~\eqref{eq:sys}), to simulate a function $g$ that accepts random inputs $w_t$ distributed according to distribution $\mathcal{W}$. Note that this is without loss of generality; we do not require $f$ to be continuous and we do not require it to be differentiable, and therefore we can represent any distribution using the inverse CDF method, as is standard in stochastic simulation \cite[Chapter 2, p.36]{asmussen2007stochastic}.
Our model is even more general, as it 
encompasses state-dependent distributions, 
whose parameters depend on the state and may depend on other distributions and thus define joint, multi-variate 
and hierarchically-structured distributions. Notably, our model comprises both continuous and discrete probability distributions and is able to model discrete-time stochastic hybrid systems~\cite{LSAZ22}. Our model also subsumes probabilistic models and programs over finite state spaces, where each state can be represented as a point of the Euclidean space $\Real^n$. Note that this encoding of a finite state space ensures Borel-measurability of the update function, as the $\sigma$-algebra of this finite state space (i.e., its powerset) is a subset of the Borel $\sigma$-algebra.

\section{Quantitative Reachability Verification of Probabilistic Models}\label{sec:reachability}

Quantitative verification treats the question of providing a quantity (probability) with which a system satisfies a property, as opposed to providing a definite (positive or negative) answer. In fact, it is sometimes too conservative or inappropriate to demand a definite outcome to  
a formal verification question. For instance, a system whose behaviour is probabilistic may violate a specification on rare corner cases and yet satisfy it with a probability that is deemed acceptable for the application domain. We address the quantitative verification of probabilistic systems, which is the problem of computing the probability for which a system satisfies a specification~\cite{DBLP:conf/sigsoft/Kwiatkowska07,DBLP:conf/icalp/BreugelW01,DBLP:journals/cacm/BaierHHK10,DBLP:reference/mc/BaierAFK18}. 

We consider the {\em quantitative reachability verification} question, which as we show below, is at the core of a variety of quantitative verification problems for probabilistic programs and stochastic dynamical models. Henceforth, we use $\mathbf{1}_S$ 
to denote the indicator function of set $S$, i.e., 
$\mathbf{1}_S(s) = 1$ if $s \in S$ and $\mathbf{1}_S(s) = 0$ if $s \not\in S$; we also use $\lambda x.M$ to denote the anonymous function that takes an argument $x$ and evaluates the expression $M$ to produce its result.

We now provide a recursive characterisation for the probability that a stochastic process reaches a Borel measurable target set in exactly time $t$ (\cref{lem:standard-lemma}), and in at most time $t$ (\cref{lem:reachstep}). Finally, in \cref{lem:reachfinlimit} we employ these results to show that the probability of reaching a target set in \textit{any} finite time is the limit of the time-bounded reachability probability as the time bound goes to infinity.

\begin{lemma}\label{lem:standard-lemma}
Let the event that a stochastic process $\{ X^s_t\}_{t \in \Nat}$ 
over state space $\Real^n$ 
initialised in state $s \in \Real^n$ 
reaches a target set $A \in \mathcal B(\Real^n)$ 
in exactly time $t \in \Nat$ be 
\begin{equation}\label{eqn:reach1}
    \Reach_t^{s}(A) = \left\{ \omega \in \Omega \mid X_0^{s}(\omega) \not\in A , \ldots, X_{t-1}^{s}(\omega) \not\in A , X_{t}^{s}(\omega) \in A \right\}, 
\end{equation}
with $\Reach_0^{s}(A) = \Omega$ if $s \in A$, and $\Reach_0^{s}(A) = \emptyset$ if $s \notin A$.
Then, $\Reach_t^{s}(A)$ is measurable and its probability measure can be expressed as follows:
\begin{subequations}
        \begin{align}
            \mathbb{P}[\Reach_{{t+1}}^{s}(A)]
            &= 
            \mathbf{1}_{\mathbb{R}^n \setminus A}(s) 
            \cdot 
            \mathbb{X}[\lambda s' \mathpunct{.} 
                \mathbb{P}[\Reach_{{t}}^{s'}(A)]
            ](s), 
            \label{eqn:standard-result-reach}\\
            \mathbb{P}[\Reach_{0}^{s}(A)]
            &= \mathbf{1}_A(s).
            \label{eqn:standard-result-zero}
        \end{align}
    \end{subequations}
\end{lemma}
\begin{proof}

For every $i$ we have that the pre-image of
any measurable set $R  \in \mathcal{B}(\mathbb{R}^n)$ under $X_i^{s} : (\Omega, \mathcal{F}) \to (\mathbb{R}^n, \mathcal{B}(\mathbb{R}^n))$, denoted $(X_{i}^s)^{-1}[R]$, is measurable since $X_i^{s}$ is a measurable function \cite[Section 3.4]{meyn_tweedie_glynn_2009}. It follows that $\Reach_t^{s}(A)$ is measurable because it can be expressed as the following finite intersection:
\begin{equation*}
    \Reach_t^{s}(A) = 
    \cap\{ (X_i^{s})^{-1}[\mathbb{R}^n \setminus A] ~\colon~ i = 0, \dots, t-1\}
    \cap (X_t^s)^{-1} [A].
\end{equation*}
It is a standard result (see \cite[Theorem 3.4.1]{meyn_tweedie_glynn_2009} and \cite[Theorem 2.8]{revuz1975markov}) that the probability measure of the event $\Reach_{{t}}^{s_0}(A)$ with $t > 0$ can be expressed as an integral with respect to the transition kernel of the 
Markov chain:
\begin{align*}
\mathbb{P}[\Reach_{{t}}^{s_0}(A)] &= \mathbf{1}_{\mathbb{R}^n \setminus A}(s_0)
\int_{s_1 \in \mathbb{R}^n \setminus A} \cdots
\int_{s_{t-1} \in \mathbb{R}^n \setminus A}  \int_{s_{t} \in A} 
~ \prod_{i=1}^t T(s_{i-1}, \mathrm{d}s_{i})\\
&= \int_{s_1 \in \mathbb{R}^n} \cdots \int_{s_{t} \in \mathbb{R}^n}
\left( \prod_{i=0}^{t-1} \mathbf{1}_{\mathbb{R}^n \setminus A}(s_i) \right) \mathbf{1}_{A}(s_t) \prod_{i=1}^{t} T(s_{i-1}, \mathrm{d}s_{i})
\end{align*}
As a consequence, we can express the probability measure of 
$\Reach_{t+1}^{s_0}(A)$ as follows:
\begin{multline*}
\mathbb{P}[\Reach_{t+1}^{s_0}(A)] = \mathbf{1}_{\mathbb{R}^n \setminus A}(s_{0})
\int_{s_1 \in \Real^n} 
\\
\underbrace{\int_{s_2 \in \Real^n} \cdots \int_{s_{t+1} \in \Real^n}
\left( \prod_{i=1}^t \mathbf{1}_{\mathbb{R}^n \setminus A}(s_i) \right)
\mathbf{1}_{A}(s_{t+1}) 
\left(\prod_{i=2}^{ t+1 } T(s_{i-1}, \mathrm{d}s_{i}) \right)}_{\mathbb{P}[\Reach_{t}^{s_1}(A)]}
T(s_0, \mathrm{d}s_1)
\end{multline*}
Then, after substituting $s_0$ with $s$ and $s_1$ with $s'$, we derive Eq.~\eqref{eqn:standard-result-reach} as follows:
\begin{align*}
\mathbb{P}[\Reach_{{t+1}}^{s}(A)]
&= 
\mathbf{1}_{\mathbb{R}^n \setminus A}(s) 
\cdot 
\int_{\Real^n}
\mathbb{P}[\Reach_{{t}}^{s'}(A)] ~T(s, \mathrm{d}s')
& \\
&= 
\mathbf{1}_{\mathbb{R}^n \setminus A}(s) 
\cdot 
\mathbb{X}[\lambda s' \mathpunct{.} 
    \mathbb{P}[\Reach_{{t}}^{s'}(A)]
](s)
&\text{by \eqref{eqn:postex}}
\end{align*}
Finally, Eq.~\eqref{eqn:standard-result-zero}, which states that $\mathbb{P}[\Reach_{0}^{s}(A)] = \mathbf{1}_A(s)$ directly follows 
from the definition of $\Reach^s_0(A)$, as  
$\Reach_{0}^{s}(A) = \Omega$ when $s \in A$ and 
$\Reach_{0}^{s}(A) = \emptyset$ when $s \not\in A$.
\end{proof}

\begin{lemma} \label{lem:reachstep}
Let the event that a stochastic process $\{ X^s_t \}_{t \in \Nat}$ 
over state space $\Real^n$ 
initialised in state $s \in \Real^n$ 
reaches a target set $A \in \mathcal B(\Real^n)$ 
in at most time $t \in \Nat$ 
be 
\begin{equation}\label{eqn:reach-union}
    \Reach_{\leq t}^{s}(A) = \cup \{~\Reach_{i}^{s}(A) \colon 0 \leq i \leq t~\}. 
\end{equation}
Then, $\Reach_{\leq t}^{s}(A)$ is measurable and its probability measure can be expressed as follows:
\begin{subequations}
    \begin{align}
        \mathbb{P}[\Reach_{\leq t + 1}^{s}(A)]
    &= 
    \mathbf{1}_A(s) 
    + 
    \mathbf{1}_{\Real^n \setminus A}(s) 
    \cdot 
    \mathbb{X}[\lambda s' \mathpunct{.} 
    \mathbb{P}[\Reach_{\leq t}^{s'}(A)]](s), 
    \label{eqn:reachafterzero}\\
    \mathbb{P}[
        \Reach_{\leq 0}^{s}(A)
    ] &= \mathbf{1}_A(s).
    \label{eqn:reachbeforezero}
    \end{align}
\end{subequations}
\end{lemma}
\begin{proof}
Since the sets in the family $\{ \Reach^{s}_{t}(A) \colon t \in \Nat\}$ are disjoint, we note that the union in \eqref{eqn:reach-union} that defines $\Reach^{s}_{\leq t}(A)$ is a union of disjoint sets. It follows that, by countable additivity of measure and disjointness of the union in \eqref{eqn:reach-union}, its probability measure is\linebreak given by  
\begin{equation}\label{eqn:reach-sum}
    \mathbb{P}[\Reach_{\leq t}^{s}(A)] = 
    \sum_{i=0}^t  \mathbb{P}[\Reach_{t}^{s}(A)].
\end{equation}
Consequently and thanks to \cref{lem:standard-lemma}, we have
    \begingroup
    \allowdisplaybreaks
    \begin{align*}
    \mathbb{P}[\Reach_{\leq {t+1}}^{s}(A)] 
    &= \mathbb{P}[\Reach_{{0}}^{s}(A)] + \sum_{i = 1}^{t + 1} \mathbb{P}[\Reach_{{i}}^{s}(A)] 
    &\text{by~\eqref{eqn:reach-sum}}\\
    &= \mathbf{1}_{A}(s) + \sum_{i = 1}^{t + 1} \mathbb{P}[\Reach_{{i}}^{s}(A)] &\text{by~\eqref{eqn:standard-result-zero}}\\
    &= \mathbf{1}_{A}(s) + \sum_{i = 0}^{t} \mathbb{P}[\Reach_{{i + 1}}^{s}(A)]\\
    &= 
    \mathbf{1}_{A}(s) + \sum_{i = 0}^{t}
    \mathbf{1}_{\mathbb{R}^n \setminus A}(s) \cdot 
    \mathbb{X}[\lambda s' \mathpunct{.} 
    \mathbb{P}[\Reach_{{i}}^{s'}(A)]](s)
    &\text{by \eqref{eqn:standard-result-reach}}\\
    &= 
    \mathbf{1}_{A}(s) + \sum_{i = 0}^{t}\left(
    \mathbf{1}_{\mathbb{R}^n \setminus A}(s) \cdot 
\int_{\Real^n}
\mathbb{P}[\Reach_{{i}}^{ s^\prime }(A)] ~T(s, \mathrm{d}s')
\right) &\text{by \eqref{eqn:postex}}\\
&= 
    \mathbf{1}_{A}(s) +
    \mathbf{1}_{\mathbb{R}^n \setminus A}(s) \cdot 
\int_{\Real^n}
\left( \sum_{i = 0}^{t} \mathbb{P}[\Reach_{{i}}^{s'}(A)] \right) ~T(s, \mathrm{d}s')
\\
&= 
    \mathbf{1}_{A}(s) +
    \mathbf{1}_{\mathbb{R}^n \setminus A}(s) \cdot 
\int_{\Real^n}
\mathbb{P}[\Reach_{{\leq t}}^{s'}(A)] ~T(s, \mathrm{d}s')
&\text{by~\eqref{eqn:reach-sum}}\\
&=\mathbf{1}_A(s) 
    + 
    \mathbf{1}_{\Real^n \setminus A}(s) 
    \cdot 
    \mathbb{X}[\lambda s' \mathpunct{.} 
    \mathbb{P}[\Reach_{{\leq t}}^{s'}(A)]](s)
    &\text{by~\eqref{eqn:postex}}.
    \end{align*}\endgroup
Furthermore, we also have that 
\begin{align*}
\mathbb{P}[\Reach_{\leq 0}^{s_0}(A)]
&= \mathbb{P}[\Reach_{0}^{s_0}(A)]
= \mathbf{1}_{A}(s_0)
&\text{by~\eqref{eqn:reach-sum} and \eqref{eqn:standard-result-zero}}. \tag*{\qedhere}
\end{align*}
\end{proof}

\begin{lemma}\label{lem:reachfinlimit}
Let the event that a stochastic process $\{ X^s_t \}_{t \in \Nat}$ 
over state space $\Real^n$ 
initialised in state $s \in \Real^n$ 
reaches a target set $A \in \mathcal B(\Real^n)$ 
in finite time be 
\begin{equation}\label{eqn:reachfin}
    \Reach_{\mathrm{fin}}^s(A) = \cup \left\{\Reach_{ t}^{s}(A) \colon t \in \Nat\right\} = 
    \cup \left\{\Reach_{\leq t}^{s}(A) \colon t \in \Nat\right\}.
\end{equation}
Then, $\Reach_{\mathrm{fin}}^s(A)$ is measurable and its probability measure (i.e.\ the reachability probability of the target set) may be expressed as follows:
    \begin{equation}\label{eqn:reachfinlim}
        \mathbb P [\Reach_{\mathrm{fin}}^s(A)] = \lim_{t \to \infty}\mathbb P [\Reach_{\leq t}^s(A)]. 
    \end{equation}
\end{lemma}
\begin{proof}
We recall from Eq.~\eqref{eqn:reach-union} that $\Reach_{\leq t}^{s}(A) = \cup \left\{\Reach_{i}^{s}(A) \colon 0 \leq i \leq t \right\}$, 
and we note that this implies that
$\Reach_{\leq t}^{s}(A) \subseteq \Reach_{\leq t+1}^{s}(A)$. 
Then, we recall from Eq.~\eqref{eqn:reachfin} that $\Reach_{\mathrm{fin}}^s(A) = \cup \left\{\Reach_{\leq t}^{s}(A) \colon t \in \Nat\right\}$. 
Finally, we apply a standard result in measure theory regarding the measure of an increasing union \cite[Theorem 2.59]{axler2019measure} to conclude \
$\mathbb P [\Reach_{\mathrm{fin}}^s(A)] = \lim_{t \to \infty}\mathbb P [\Reach_{\leq t}^s(A)]$.\qedhere
\end{proof}

The recursive characterisation of the time-bounded reachability probability (\cref{lem:reachstep}) and the limit characterisation of the probability of reaching the target set in any finite time (\cref{lem:reachfinlimit}) provides the foundation for a fixpoint characterisation for the reachability probability (\cref{sec:fixpoint-characterisation}), in turn justifying the conditions that a supermartingale certificate must satisfy (\cref{sec:neuralism}) for it to provide a sound upper bound on the reachability probability.

Our method leverages neural networks to compute an upper bound 
for the reachability probability with respect to a given target set (cf.\ \cref{sec:neuralism}).
Notice that an appropriate choice of target set allows us to express a variety of other quantitative verification questions,
for which our method can provide an upper or lower bound. 
Next, by providing a suitable choice for the target set, we show how important quantitative verification questions for probabilistic programs and stochastic dynamical models can be characterised as instances of quantitative reachability.

\begin{description}
\item[Termination Analysis]
Let $G \in \mathcal B (\Real^n)$ be the guard set of a probabilistic while loop as in Eq.~\eqref{eq:loop}, i.e., the set of states for which 
that guard condition evaluates to true. 
The event that the loop terminates from initial state $s_0$ is $\Reach_{\mathrm{fin}}^{s_0}(\Real^n \setminus G)$. 
Our method computes an upper bound for the probability that the loop terminates or, dually, it computes a lower bound for the probability of non-termination. Notably, when this lower bound is greater than 0, then almost-sure termination is refuted.
\item[Assertion-violation Analysis]
Let $G \in \mathcal B (\Real^n)$ be the guard set of a 
probabilistic while loop and $A \in \mathcal B (\Real^n)$ be the satisfying set of an assertion placed at the beginning of the 
loop body. Given initial state $s_0$, the event that the assertion is eventually violated 
is $\Reach_{\mathrm{fin}}^{s_0}(G \setminus A)$. 
Our method computes an upper bound for the probability of assertion violation or, dually, a lower bound for its satisfaction. Note that  assertions in other positions of the body can be modelled similarly by using additional Boolean variables.
\item[Safety Verification] 
Let $B \in \mathcal B (\Real^n)$ be a set of undesirable states in a stochastic dynamical model. The event that the system is safe when initialised in $s_0$, i.e., it never reaches an undesirable state, 
is given by 
$\Omega \setminus \Reach_{\mathrm{fin}}^{s_0}(B)$. 
Our method computes a lower bound for the probability that the system is safe, which is $1 - \mathbb P [\Reach_{\mathrm{fin}}^{s_0}(B)]$.
\item[Invariant Verification] Let $I \in \mathcal B (\Real^n)$ be a candidate invariant set. The event that $I$ is invariant when the system is initialised in $s_0$ is $\Omega \setminus \Reach_{\mathrm{fin}}^{s_0}(\Real^n \setminus I)$.
Our method computes a lower bound for the probability 
that set $I$ is invariant, which is $1 - \mathbb P [\Reach_{\mathrm{fin}}^{s_0}(\Real^n \setminus I)] $. Note that if $s_0 \notin I$, this definition yields a trivial lower bound of zero for the probability of invariance.

\end{description}

\section{Fixpoint Characterisation of Quantitative Reachability}\label{sec:fixpoint-characterisation}


Our supermartingale certificate proof-rule (\cref{thm:V-of-s})
relies on the order-theoretic characterisation of reachability probabilities 
in probabilistic programs \cite{DBLP:conf/atva/TakisakaOUH18}. 
This previous seminal work considered a more general setting allowing 
for both non-determinism and arbitrary program invariants.   
For clarity and self-containment, and to build intuition about 
order-theoretic reasoning in our context, we present a streamlined analysis that mirrors the soundness proof for non-negative supermartingales \cite[Proposition 4.2]{DBLP:conf/atva/TakisakaOUH18}
in a setting without non-determinism or invariants.

We define the space of candidate certificates as the set of $({\mathcal{B}}(\Real^n), \mathcal{B}(\Real^n_{\geq 0} \cup \{ \infty \}))$-measurable functions $\mathcal{V} \subset \Real^n \to \Real_{\geq 0} \cup \{ \infty \}$. Furthermore, we define the ordering relation $\leq$ between two functions $U,V \in \mathcal{V}$ pointwise as follows:
\begin{equation}\label{eqn:order}
     U \leq V \Leftrightarrow \forall s \in \Real^n \mathpunct{.} U(s) \leq V(s).
\end{equation}
The set of functions $\mathcal{V}$ under the ordering $\leq$ defines a 
pointed complete partially ordered set (pointed CPO) \cite[Definition 8.1]{davey_priestley_2002} with the constant zero function as its bottom element; 
we overload the symbol $\mathbf{0} \in \mathcal{V}$ to denote the constant zero function. 

\begin{lemma}\label{lem:cpo}
$(\mathcal{V}, \leq)$ is a pointed CPO with bottom element $\mathbf{0}$, and with each countable increasing chain $\{V_i \}_{i \in \mathbb{N}}$ having least upper bound:
\begin{equation}\label{eqn:lublim}
(\sup~ \{V_i \colon i \in \Nat \})(s)
            = \sup~\{ V_i(s) : i \in \Nat \} 
            = \lim_{i \to \infty} V_i(s).
\end{equation}
\end{lemma}
\begin{proof}

The fact that the pointwise relation $\leq$ on functions is reflexive, transitive and anti-symmetric follows from the properties of the natural ordering relation $\leq$ on the set $\mathbb{R}_{\geq 0} \cup \{ \infty \}$. This shows that $(\mathcal{V}, \leq)$ is a partial order. The constant function $\mathbf{0}$ is the bottom element of the partial order because $\mathcal{V}$ contains only non-negative functions.

To show that the partial order is also complete, we need to show that the least upper bound of every countable increasing chain $\{V_i \}_{ i \in \mathbb{N} }$ in $(\mathcal{V}, \leq)$ is also in $\mathcal{V}$; in other words, that the least upper bound is $(\mathcal{B}(\Real^n), \mathcal{B}(\Real_{\geq 0} \cup \{ \infty \}))$-measurable. To this end, we note that for all $s \in \Real^n$ we have $(\sup~\{ V_i \colon i \in \Nat \})(s) = \sup ~ \{ 
V_i(s) \colon i \in \Nat
\}$ owing to a standard result \cite[Lemma 8.5(ii), p.177]{davey_priestley_2002}, and that the pointwise supremum of a countable family of measurable functions is also measurable {\protect\cite[Example 8, p.23]{pollard_2001}}. 

Finally, since for any $s \in \Real^n$ we have $\{V_i(s)\}_{i \in \mathbb{N}}$ is a non-decreasing sequence in the extended non-negative reals (namely, the set $\mathbb{R}_{\geq 0} \cup \{ \infty \}$), we have by the monotone convergence property of non-decreasing real sequences \cite[Theorem 15,  p.21]{royden2010real} that
    \begin{equation*}
        \lim_{i \to \infty} V_i(s) = \sup~\{ V_i(s) \colon i \in \mathbb{N}\}.\qedhere
    \end{equation*}
\end{proof}

The exact reachability probability is a measurable function from states to non-negative reals, and is thus an element of the pointed CPO $\mathcal{V}$. By \cref{lem:reachstep} we have a recursive characterisation of increasingly tight approximations to this probability in terms of the higher-order endofunction
$\Phi : \mathcal{V} \to \mathcal{V}$ defined as
\begin{equation}\label{def:phi}
    \Phi(V)(s) = 
    \mathbf{1}_{A}(s) + \mathbf{1}_{\Real^n \setminus A}(s) \cdot \mathbb{X}[V](s).
\end{equation} 
Furthermore, \cref{lem:reachfinlimit} states that the exact reachability probability is the limit of its approximations.
We formalise this argument by an application of Kleene's Fixpoint Theorem (see \cite[Section 8.15, p.183]{davey_priestley_2002} and \cite[Theorem 5.11]{DBLP:books/daglib/0070910}) over the pointed CPO $(\mathcal{V}, \leq)$ over the space of certificates. A technical requirement of Kleene's Fixpoint Theorem 
is that $\Phi$ is $\omega$-continuous \cite[Definition p.71]{DBLP:books/daglib/0070910}. A function is $\omega$-continuous if it is monotonic and preserves least upper bounds of countable increasing chains. 
To show that the higher-order function $\Phi$ is $\omega$-continuous, we first show that the post-expectation operator ${\mathbb{X}}$ (which is in turn a higher-order function in $\mathcal{V}$) is $\omega$-continuous.

\begin{lemma}\label{lem:nexttime-cont}
$\mathbb{X}$ is $\omega$-continuous.
\end{lemma}
\begin{proof}
Recalling the definition of $\omega$-continuity for an endofunction on a pointed CPO \cite[Definition p.71]{DBLP:books/daglib/0070910}, we show that $\mathbb{X}$ satisfies the following two properties:
\begin{itemize}
    \item $\mathbb{X}$ {\em is monotonic}. Suppose $U,V \in \mathcal{V}$ with $U \leq V$. Then for all $s \in \mathbb{R}^n$:
    \begin{align*}
        \mathbb{X}[U](s) &= \int_{\mathbb{R}^n} U(s^\prime) ~T(s, \mathrm{d}s^\prime)\\
        &\leq 
        \int_{\mathbb{R}^n}
        V(s^\prime) ~T(s, \mathrm{d}s^\prime)
        &&\text{by $U \leq V$}\\
        &= \mathbb{X}[V](s).
    \end{align*}

    \item $\mathbb{X}$ {\em preserves least upper bounds of countable increasing chains}. Let $V_0 \leq V_1 \leq \ldots $ be a countable increasing chain in $\mathcal{V}$. As $\mathbb{X}$ is monotone, $\mathbb{X}[V_0] \leq \mathbb{X}[V_1] \leq \ldots$ is also a countable increasing chain. 
    We note that since ${\mathbb{X}}$ is an integral, it satisfies the Monotone Convergence Theorem (MCT) \cite[Theorem 12]{pollard_2001}, which can be invoked in the following since $\{V_i\}_{i \in \mathbb{N}}$ is a sequence of non-negative measurable functions that converge monotonically to the function $\sup~\{ V_i \colon i \in \mathbb{N}\}$.
    Thus, for all $s \in \mathbb{R}^n$, we have that
    \begin{align*}
        (\sup~ \{ 
        \mathbb{X}[V_i] \colon i \in \mathbb{N}
        \})(s)
        &= \lim_{i \to \infty} \mathbb{X}[V_i](s)
        &\text{by~\eqref{eqn:lublim}, \cref{lem:cpo}}\\
        &= \lim_{i \to \infty} 
        \int_{\Real^n} V_i(s^\prime)~T(s, \mathrm{d}s^\prime)
        &\text{by \eqref{eqn:postex}}
        \\
        &= \int_{\Real^n} (\sup~\{ V_i \colon i \in \mathbb{N} \})(s^\prime)~T(s, \mathrm{d}s^\prime)
        &\text{\eqref{eqn:lublim} \& MCT}
        \\
        &= \mathbb{X}[\sup~\{ V_i \colon i \in \mathbb{N} \}](s) \tag*{\qedhere}
    \end{align*}
\end{itemize}
\end{proof}
\begin{lemma}\label{lem:scott}
$\Phi$ is $\omega$-continuous.
\end{lemma}
\begin{proof}
We show that $\Phi$ satisfies the following two properties.
\begin{itemize}
\item $\Phi$ is {\em monotonic}.
    Let $s \in \Real^n$ be arbitrary and assume $U \leq V$. Then
        \begin{align*}
            \Phi(U)(s) &= 
            \mathbf{1}_{A}(s) + \mathbf{1}_{\Real^n \setminus A}(s) \cdot \mathbb{X}[U](s)\\
            &
            \leq \mathbf{1}_{A}(s) + \mathbf{1}_{\Real^n \setminus A}(s) \cdot \mathbb{X}[V](s) &\text{by monotonicity of $\mathbb{X}$, \cref{lem:nexttime-cont}}\\
            &= \Phi(V)(s),
        \end{align*}
    therefore $\Phi(U) \leq \Phi(V)$.

\item $\Phi$ {\em preserves least upper bounds of countable increasing chains}. Suppose $\{ V_i \colon i \in \mathbb{N}\}$ is a countable increasing chain in $(\mathcal{V}, \leq)$. 
Let $s \in \Real^n$ be arbitrary. Then
        \begin{align*}
            \Phi(\sup~ \{V_i \colon i \in \Nat \})(s)
            &= 
            \mathbf{1}_{A}(s) + \mathbf{1}_{\Real^n \setminus A}(s) \cdot \mathbb{X}[
            \sup~ \{V_i \colon i \in \Nat \}
            ](s) &\text{by \eqref{def:phi}}\\
            &= 
            \mathbf{1}_{A}(s) + \mathbf{1}_{\Real^n \setminus A}(s) \cdot 
            (\sup~ \{\mathbb{X}[V_i] \colon i \in \Nat \}) (s) &\text{by \cref{lem:nexttime-cont}}\\
            &= 
            \mathbf{1}_{A}(s) + \mathbf{1}_{\Real^n \setminus A}(s) \cdot 
            (\sup~\{\mathbb{X}[V_i](s) \colon i \in \Nat \}) &\text{by \eqref{eqn:lublim}}\\
            &= 
            \sup~ \{
            (\mathbf{1}_{A}(s) + \mathbf{1}_{\Real^n \setminus A}(s)
            \cdot \mathbb{X}[V_i](s)) \colon i \in \Nat
            \}\\
            &= 
            \sup~ \{
            \Phi[V_i](s) \colon i \in \Nat
            \} &\text{by \eqref{def:phi}}\\
            &= \left(\sup~ \{ \Phi(V_i) \colon i \in \Nat \}\right)(s)
            &\text{by \eqref{eqn:lublim};}
        \end{align*}
we note that the steps above relying on Eq.~\eqref{eqn:lublim} use the property that suprema in the pointed CPO $(\mathcal{V}, \sqsubseteq)$ are defined pointwise.
Finally, the derivation above establishes that $\Phi(\sup~ \{ V_i \colon i \in \mathbb{N}\}) = \sup~ \{ \Phi(V_i) \colon i \in \mathbb{N} \}$.
\qedhere
\end{itemize}
\end{proof}

To demonstrate that the exact reachability probability is the least fixpoint of the higher-order function $\Phi$, we leverage Kleene's Fixpoint Theorem which states that since $\Phi$ is $\omega$-continuous, its least fixpoint is given by 
\begin{equation}
    \sup~\{ \Phi^t (\mathbf{0}) \colon t \in \Nat \}
\end{equation}
where we let $\Phi^t$ denote the $t$-fold application of $\Phi$ to its argument.
To relate this to quantitative reachability, we show in \cref{lem:reachlublambda} that for any $t \in \Nat$, the function $\lambda s \mathpunct{.} \mathbb{P}[\Reach_{\leq t}^s(A)]$ is equal to $\Phi^{t+1}(\mathbf{0})$ (which is simply a re-statement of \cref{lem:reachstep}). In combination with Eq.~\eqref{eqn:reachfinlim}, this yields the following result. 
\begin{lemma}\label{lem:reachlublambda}
    $\lambda s \mathpunct{.} \mathbb{P}[\Reach_{\mathrm{fin}}^{s}(A)]$ is the least fixpoint of $\Phi$.
\end{lemma}
\begin{proof}
We need to prove
\begin{equation}
    \lambda s . \mathbb{P}[\Reach_{\mathrm{fin}}^{s}(A)]
    = (\sup~ \{ \Phi^t(\mathbf{0}) \colon t \in \mathbb{N}\}).
    \label{eqn:reachlublambda}
\end{equation}
To prove Eq.~\eqref{eqn:reachlublambda} we first show that for all $t \in \Nat$ and $s \in \Real^n$, we have
    \begin{equation}
    \mathbb{P}[\Reach_{\leq t}^{s}(A)]
    = \Phi^{t+1}(\mathbf{0})(s).\label{eqn:kleene-1}
    \end{equation}
    We proceed by induction on $t \in \Nat$, that is, we show
    \begin{itemize}
    \item $\lambda s \mathpunct{.} \mathbb{P}[\Reach_{\leq 0}^{s}(A)]
    = \Phi^{1}(\mathbf{0})$, and 
    \item 
    $\lambda s \mathpunct{.} \mathbb{P}[\Reach_{\leq t}^{s}(A)]
    = \Phi^{t+1}(\mathbf{0})$
    implies
    $\lambda s \mathpunct{.} \mathbb{P}[\Reach_{\leq t+1}^{s}(A)]
    = \Phi^{t+2}(\mathbf{0})$.
    \end{itemize}
    For the base case, we let $s \in \Real^n$ and have
    \begin{align*}
\mathbb{P}[\Reach_{\leq 0}^{s}(A)]  
&= \mathbf{1}_A(s) &\text{by \eqref{eqn:reachbeforezero}}\\
&= \mathbf{1}_A(s) + \mathbf{1}_{\Real^n \setminus A}(s) \cdot \mathbb{X}[\mathbf{0}](s)
&\text{because $\mathbb{X}[\mathbf{0}] = \mathbf{0}$}\\
&= \Phi^1(\mathbf{0})(s) &\text{by \eqref{def:phi}.}
    \end{align*}
For the inductive case, we let $s \in \Real^n$ and leverage \cref{lem:reachstep} as follows:
\begin{align*}
    \mathbb{P}[\Reach_{\leq t + 1}^{s}(A)]
    &= 
    \mathbf{1}_A(s) 
    + 
    \mathbf{1}_{\Real^n \setminus A}(s) 
    \cdot 
    \mathbb{X}[
    \lambda s^\prime . 
    \mathbb{P}[\Reach_{\leq t}^{s^\prime}(A)]
    ](s)
    &\text{by \eqref{eqn:reachafterzero}}    \\
    &= \Phi(\lambda s^\prime .  \mathbb{P}[\Reach_{\leq t}^{s^\prime}(A)])(s) &\text{by \eqref{def:phi}}\\
    &= \Phi(\Phi^{t+1}(\mathbf{0}))(s)&\text{by IH}\\
    &= \Phi^{t+2}(\mathbf{0})(s)
\end{align*}

Next, we show that \eqref{eqn:reachlublambda} is a consequence of \cref{lem:cpo}, noting that $\{\Phi^t(\mathbf{0}) \}_{ t \in \mathbb{N}}$ is a countable increasing chain in  $(\mathcal{V},\leq)$:
\begin{align*}
    \mathbb{P}[\Reach_{\mathrm{fin}}^{s}(A)] &= \lim_{t \to \infty} \mathbb{P}[\Reach_{\leq t}^{s}(A)] &\text{by \eqref{eqn:reachfinlim}}\\
    &= \lim_{t \to \infty} \Phi^{t + 1}(\mathbf{0})(s) &\text{by \eqref{eqn:kleene-1}}\\
    &= \lim_{t \to \infty} \Phi^{t}(\mathbf{0})(s) \\
    &= (\sup~\{ \Phi^{t}(\mathbf{0}) \colon t \in \Nat \})(s) &\text{by~\eqref{eqn:lublim}, \cref{lem:cpo}} &\qedhere
\end{align*}
\end{proof}

In summary, \cref{lem:reachlublambda} characterises the exact reachability probability as the least fixpoint of a $\omega$-continuous endofunction $\Phi$ on the set $\mathcal{V}$. We use this fixpoint characterisation in \cref{sec:neuralism} to show how a supermartingale certificate establishes an upper bound on the probability of reaching a target state in any number of steps.
    
\section{Neural Supermartingales for Quantitative Verification}\label{sec:neuralism}

Supermartingale certificates provide a flexible and powerful 
theoretical framework for the formal verification of
probabilistic models with infinite state spaces (see \cite[Section 5.2]{Durrett2010}, \cite[Definition 9 and Theorem 4]{DBLP:conf/cav/ChakarovS13}).
Not only have supermartingales been applied to qualitative questions,  
such as almost-sure termination analysis of probabilistic programs 
and almost-sure stability analysis of stochastic dynamical models, 
but also to quantitative reachability verification, as in this work.   
Specifically, this is enabled by the theory of non-negative repulsing supermartingales and of stochastic invariants~\cite{DBLP:journals/toplas/TakisakaOUH21,DBLP:conf/cav/ChatterjeeGMZ22,DBLP:conf/popl/ChatterjeeNZ17}. 

Our method builds upon the theory of non-negative repulsing supermartingales and stochastic invariant indicators 
and applies it to take advantage of the expressive power of neural networks and the 
flexibility of machine learning and counterexample-guided inductive synthesis algorithms. We refer to a function $V$ that is non-increasing in expectation outside of the target set as a \textit{supermartingale}. 
This is because $V(X_t^s)$ defines a supermartingale process \cite[Section 5.2]{Durrett2010}, where $X_t^s$ is the stochastic process corresponding to the semantics of the probabilistic model, c.f.\ Eq.~\eqref{eqn:stochprocess}. 
Our method hinges on the following theorem.

\begin{theorem}\label{thm:V-of-s}
Let $A \in \mathcal{B}(\Real^n)$ be a target set, and let $V \colon \Real^n \to \Real_{\geq 0} \cup \{\infty\}$ be a non-negative function that satisfies the following two conditions:
\begin{subequations}
\label{eq:spec}
\begin{flalign}    
     &\text{(indicating condition)} &&\forall s \in A \colon V(s) \geq 1, && \label{eq:spec_ind}\\ 
     &\text{(non-increasing condition)} &&\forall s \not\in A \colon \mathbb{X}[V](s) \leq V(s),&&\phantom{aaaaaaaa}  \label{eq:spec-dec}
\end{flalign}
\end{subequations}
Then for all $s \in \Real^n$ we have $V(s) \geq \mathbb P [\Reach_{\mathrm{fin}}^{s}(A)]$.
\end{theorem}
\begin{proof}

Suppose that $V \in \mathcal{V}$ is a function that satisfies  \eqref{eq:spec_ind} and \eqref{eq:spec-dec}. We show that $V$ is a pre-fixpoint of $\Phi$ (defined in Eq.~\ref{def:phi}), namely, that $\Phi(V) \leq V$.
We proceed by cases:
\begin{align*}
    &\text{Case $s \in A$: }~\Phi(V)(s)= 
    \mathbf{1}_{A}(s) + \mathbf{1}_{\Real^n \setminus A}(s) \cdot \mathbb{X}[V](s) = 1 \leq V(s) &\text{by \eqref{eq:spec_ind}},\\
    &\text{Case $s \notin A$: }~
    \Phi(V)(s)= 
    \mathbf{1}_{A}(s) + \mathbf{1}_{\Real^n \setminus A}(s) \cdot \mathbb{X}[V](s)
    = \mathbb{X}[V](s)
    \leq V(s) &\text{by \eqref{eq:spec-dec}}.
\end{align*}    
This establishes that $V$ is a pre-fixpoint of $\Phi$.

Next, we recall the standard result \cite[Section 8.20, p.186]{davey_priestley_2002} that (if it exists) the least fixpoint of a monotone endofunction is equal to its least pre-fixpoint (and therefore, the least fixpoint is less than or equal to \textit{any} pre-fixpoint). In our case, since the endofunction $\Phi : \mathcal{V} \to \mathcal{V}$ is $\omega$-continuous, its least fixpoint exists and is equal (by \cref{lem:reachlublambda}) to the exact reachability probability, $\lambda s \mathpunct{.} \mathbb{P}[\Reach_{\mathrm{fin}}^{s}(A)]$. Thus, as $V$ is a pre-fixpoint of $\Phi$, we have that $\lambda s \mathpunct{.} \mathbb{P}[\Reach_{\mathrm{fin}}^{s}(A)] \leq V$ holds true, 
which corresponds to the statement of \cref{thm:V-of-s}.\qedhere
\end{proof}

As we show in detail in \cref{sec:cegis}, we compute a supermartingale that satisfies the two criteria \eqref{eq:spec_ind} and \eqref{eq:spec-dec},
while also minimising its output over the initial state $s_0$. When the initial state is chosen nondeterministically from a set $S_0$, we instead minimise the output over all states in the set. As a consequence of \cref{thm:V-of-s}, the maximum of $V$ over $S_0$ is a sound upper bound for the reachability probability.

In this work, we template $V$ as a neural network with $n$ input neurons, $l$ hidden layers with respectively 
$h_1, \dots h_l$ neurons in each hidden layer, and 1 output neuron.
We guarantee a-priori that the function's output will be non-negative over the entire domain using the following architecture:
\begin{equation}
\label{eq:net}
    V(x) = \operatorname{sum} \left( \sigma_l \circ \dots \circ \sigma_1(x) \right), \quad 
    \sigma_i(z) = 
    (\relu(\mathbf{w}^T_{i,1} z + \mathbf{b}_{i,1}), \dots, \relu(\mathbf{w}^T_{i,h_i} z + \mathbf{b}_{i,h_i})), 
\end{equation}
where $\operatorname{sum}(z_{l, 1}, \ldots, z_{l, h_l}) = \sum_{k = 1}^{h_l} z_{l, k}$, $\relu{}(z) = \max\{0, z\}$, and $\mathbf{w}^T_{i,j} \in \mathbb{Q}^{h_{i-1}}$ (defining $h_0 = n$, the number of input neurons) and 
$\mathbf{b}_{i,j} \in \mathbb{Q}$ are respectively the weight vector and the bias parameter for the inputs to neuron $j$ at layer $i$. 
Then, our method trains the neural network to satisfy the two criteria \eqref{eq:spec_ind} and \eqref{eq:spec-dec}, 
while also minimising its output on the initial condition. 

\begin{figure}
\centering
    \begin{tabular}{ccc}
    \begin{minipage}[b]{0.4\linewidth}
\begin{lstlisting}
while x > 0 do
    assert(x <= 10);
    p ~ Bernoulli(0.5);
    if p == 1 then
        x -= 2
    else
        x += 1
    fi
od
\end{lstlisting}
\end{minipage}
    &\quad &
    \begin{tikzpicture}[outer sep=0, inner sep=0]
        \begin{axis}[axis x line=center,
                    axis y line=center,
                    xmin=-3,xmax=12,
                    xtick={-2, -1, 0,1,...,11},
                    ytick={0,1},height=5.5cm,width=0.55\linewidth,
                    x label style={anchor=north, below=5mm, left=17mm},
                    y label style={rotate=90,anchor=south, left=20mm, above=1mm},
                    xlabel=Initial value of {\tt x},
                    ylabel=Probability]

            \addplot [black!60, very thick] table \repulsedatastepfunc;
            \addplot [orange, very thick, domain=-3:11] {x/12+1/6};
            \addplot [blue,very thick] coordinates {(-3, 0.00) (-2,0.00) (0,0.04) (1,0.06) (2,0.09) (3,0.12) (4,0.16) (5,0.21) (6,0.27) (7,0.36) (8,0.5) (9,0.7) (10,1) (10.5, 1.15)};
            \addplot [very thick, magenta] coordinates {(-3, 0.01) (-2, 0.01) (0, 0.069) (5, 0.39) (10, 1) (10.5, 1.061)};
        \end{axis}
        \node at (2.9,2.2) {\color{orange}$V^{\text{lin}}$};
        \node at (3.0,1.55) {\color{magenta}$V^{\text{neur}}_{3}$};
        \node at (3.0,1) {\color{blue}$V^{\text{neur}}_{12}$};
        \node at (4.5, 0.8) {\color{black!60}{ $P_\text{reach}$ }};
    \end{tikzpicture}
    \end{tabular}
    \caption{Comparison between linear and neural supermartingale functions in tightness of bounds for the assertion violation probability of the program {\sf repulse}, shown on the left. On the right, the function ${V^\text{lin}}$ indicates the tightest linear supermartingale (under the restriction that {\tt x} is greater than $-2$). The functions ${V^\text{neur}_3}$ and ${V^\text{neur}_{12}}$ indicate single-layer neural supermartingales with 3 and 12 neurons respectively. The piecewise constant function $P_\text{reach}$ is the true probability of assertion violation, and indicates the ideal lower bound for the value of any supermartingale function.
    }

    \label{fig:Intuition}
\end{figure}
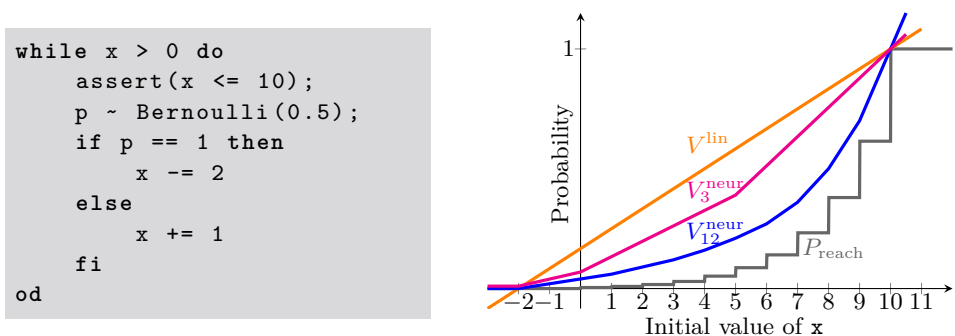

Using neural networks as templates of supermartingale functions
introduces non-trivial advantages with respect to symbolic methods 
for the synthesis of supermartingales. 
Firstly, neural supermartingales are able to better approximate the true probability of reachability and thus attain tighter upper bounds on it. 
Secondly, symbolic methods require deterministic invariants that overapproximate the set of reachable states, and to suitably restrict the domain of the template. Figure~\ref{fig:Intuition} illustrates using an example the advantages of neural certificates with respect to linear supermartingales synthesised using Farkas' lemma.  
This example shows that increasing the number of neurons provides greater flexibility and allows the certificate to more tightly approximate the true probability. 
Moreover, this example shows that linear supermartingales require their domain to be restricted with an appropriate deterministic invariant (to enforce non-negativity). 
Notably, symbolic methods for the synthesis of polynomial supermartingales based on Putinar's Positivstellensatz also require compact deterministic invariants to be provided~\cite{DBLP:conf/cav/ChatterjeeGMZ22}. By contrast, neural supermartingales (whose output is always non-negative) achieve the same result while 
relaxing the requirement of providing an invariant beforehand.

\section{Data-driven Synthesis of Neural Supermartingales}\label{sec:cegis}

Our approach to synthesising neural supermartingales for quantitative verification utilises a counterexample-guided inductive synthesis (CEGIS) procedure \cite{solar-lezama2008ProgramSynthesisSketching, DBLP:conf/asplos/Solar-LezamaTBSS06} (cf.\ \cref{fig:CEGIS}). This procedure consists of two components,
a learner and a verifier, that work in opposition to each other. 
On the one hand, the learner seeks to synthesise a candidate supermartingale that meets the desired specification (cf.\ Eq.~\eqref{eq:spec}) over a finite set of samples from the state space, while simultaneously optimising the tightness of the probability bound. On the other hand, the verifier seeks to disprove the validity of this candidate by searching for counterexamples, i.e., instances where the desired specification is invalidated, over the entire state space. 
If the verifier shows that no such counterexample exists, then the desired specification is met by the supermartingale and the procedure provides a sound probability bound for the reachability probability of interest, together with a neural supermartingale to certify it.

\begin{figure}
    \centering
    \begin{tikzpicture}
\usetikzlibrary{arrows.meta}

\node (input) at (-2, 2) {};
\node (learner) [draw, rectangle, rounded corners, minimum width=1.5cm, minimum height =0.75cm]
	at (0,2) {Learner};

\node (verif) [draw, rectangle, rounded corners, minimum width=1.5cm, minimum height =0.75cm]
	at (3,2) {Verifier};

\node (out) [label={[align=left]$V, {p}$}] at (5,2) {} ;

\draw [-{Stealth[scale=1.5]}]  (input) edge node[below, xshift=-0.2cm] {$f$, {$S_0$}, $A$} (learner);

\draw [-{Stealth[scale=1.5]}, bend left] (verif) edge node[below] {$d_\text{cex}$} (learner);

\draw [-{Stealth[scale=1.5]}, bend left] (learner) edge node[above] {\normalsize Candidate $V$} (verif);

\draw[-{Stealth[scale=1.5]}] (verif) edge (out);

\end{tikzpicture}
  
    \caption{Overview of the counterexample-guided inductive synthesis procedure used to synthesise neural supermartingales for quantitative reachability verification. Inputs to the procedure are a probabilistic program $f$, a set of initial states {$S_0$}, and a target set $A$. The procedure outputs a valid neural supermartingale $V$ and a probability bound $p$. }
    \label{fig:CEGIS}
\end{figure}
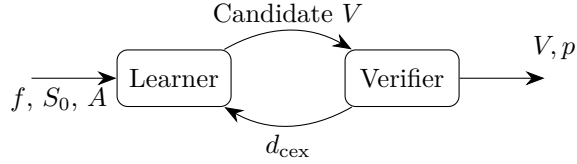

\subsection{Training of Neural Supermartingales from Samples}\label{sec:learner}

Our neural supermartingale for quantitative reachability verification consists of a neural network with $\relu$ activation functions, with an arbitrary number of hidden layers (cf.~\cref{sec:neuralism}). We train this neural network using gradient descent over a finite set $D = \{ d^{(1)}, \dots , d^{(m)}\} \subseteq \mathbb{R}^n$ of states $d^{(i)}$ sampled over the state space $\Real^n$. 
Initially, we sample uniformly within a bounded hyper-rectangle of $\mathbb{R}^n$, a technique that scales effectively to high dimensional state spaces and efficiently populates the initial dataset of state samples. Then, we construct a loss function that 
guides gradient descent to optimise the parameters (weight and biases) 
of the neural network $V$ to satisfy the specification set out in Eq.~\eqref{eq:spec} while also minimising the probability bound.
We define a loss function $\mathcal{L}(D)$ that consists of three terms:  
\begin{equation}
    \label{eq:loss}
    \mathcal{L}(D) = \beta_{1}\mathcal{L}_\indicating(D) + \beta_{2} \mathcal{L}_\nonincreasing(D) + \beta_3\mathcal{L}_{\bound}(D).
\end{equation}
Components $\mathcal{L}_\indicating$ and $\mathcal{L}_\nonincreasing$ are responsible for encouraging satisfaction of the conditions in Eq.~\eqref{eq:spec}, while the component $\mathcal{L}_{\bound}$ is responsible for tightening the probability bound. The parameters of this optimisation problem are the parameters of the neural network, which are initialised randomly. The dataset consists of the state samples, initially sampled randomly, and generated from counterexamples in subsequent CEGIS iterations (cf.\ \cref{sec:verification}). The coefficients $\beta_{1}, \beta_2$ and $\beta_3$ denote scale factors for each term, which we choose according to the priority that we want to assign to each condition (cf.\ \cref{sec:experiments}).  

First, consider the condition in Eq.~\eqref{eq:spec_ind}, which we refer to as the \textit{indicating condition}. For this, we use the following loss function:  
\begin{flalign}    
\label{eq:ind-loss}
     &\text{(indicating loss)} 
     &&\mathcal{L}_\indicating(D) = 
     \frac{1}{
     | D \cap A|
     }
     \sum_{d \in D \cap A}
     \relu{}(1 - V(d))
     &&\phantom{aaaaaaaa} 
\end{flalign}
This adds a penalty for states $d \in D$ lying inside the target condition $A$, at which $V$ fails to satisfy the indicating condition, whilst ignoring any states where $V$ satisfies it. We average this per-state penalty across all states in $D \cap A$ (where $| D \cap A|$ denotes the number of elements in this set).

We next consider the \textit{non-increasing condition} in Eq.~\eqref{eq:spec-dec}, for which we use the following loss term:
\begin{flalign}    
\label{eq:dec-loss}
     &\text{(non-increasing loss)}
     &&\mathcal{L}_\nonincreasing(D) = 
     \frac{1}{
     |D \setminus A|
     }
     \sum_{
     d \in D \setminus A
     }
     \ReLU(\mathbb{X}[V](d) - V(d)).
     &&\phantom{aaa} 
\end{flalign}
This penalises states $d \in D$ lying outside of the target set $A$, at which $V$ fails to satisfy the non-increasing condition.
Notably, this component is defined in terms of the post-expectation $\mathbb{X}[V]$ of our supermartingale. To embed this expression in 
our loss function we consider two alternative approaches, which we call \textit{program-aware} and \textit{program-agnostic}. 
\begin{description}
\item[Program-aware Approach] The program-aware approach uses the source code of the program to generate a symbolic expression for the post-expectation of $V$. For this purpose, we exploit the symbolic inference algorithm introduced by the tool PSI \cite{gehr2016psi, gehr2020lpsi}, along with a symbolic representation of $V$. We construct a probabilistic program which represents the evaluation of $V$ on the state resulting after the execution of the update function $f$. The expected value of this program's result is precisely $\mathbb{X}[V]$. This results in a symbolic expression that is a function of the program state and parameters of $V$. In the non-increasing loss, states are instantiated to elements of $D \setminus A$, while the parameters are left as free-variables that the gradient descent engine differentiates with respect to.

\item[Program-agnostic Approach] 
The program-agnostic approach provides an alternative formulation of the non-increasing loss term that does not require symbolic reasoning. Instead of directly leveraging the program's source code, it assumes that for any state we can draw independent samples from its distribution over successor states, i.e., we assume access to a  \textit{generative model} \cite{DBLP:journals/jmlr/MaillardMLG10,DBLP:conf/ijcai/KearnsMN99,DBLP:conf/icml/JinS21}. We therefore utilise a Monte Carlo scheme to estimate the post-expectation. 
For each state $d$ in our dataset $D$, to obtain an estimate of $\mathbb{X}[V](d)$ we sample a number $m^\prime$ of successor states $D^\prime = \{ {d^\prime}^{(k)}~:~1 \leq k \leq m^\prime\}$. Each successor state ${d^\prime} ^{(k)}$ is sampled by executing the program's update function $f$ (cf.\ Eq.~\eqref{eq:sys}) at state $d$. Then $\mathbb{X}[V](d)$ is estimated as 
$\mathbb{X}[V](d) \approx \frac{1}{|D^\prime|} \sum_{d^\prime \in D^\prime} V(d^\prime)$,
namely the average of $V$ over $D^\prime$.
Even though this is an approximation, we emphasise that this does not affect the soundness of our scheme, which is ensured by the verifier.
\end{description}

Finally, we introduce a tightness criterion that minimises the probability upper bound:
\begin{flalign}    
\label{eq:min-loss}
     &\text{(minimisation loss)}
     &&
     \mathcal{L}_\bound(D) = 
     \frac{1}{
     |D \cap S_0|
     }
     \sum_{d \in D \cap S_0}
     V(d).
     &&\phantom{aaa} 
\end{flalign}
This term encourages $V$ to take smaller values over the set of initial states $S_0$. Recall that the smaller the probability upper bound, the closer it is to the true value.
 
The loss function is provided to the gradient descent optimiser, whose performance benefits from a smooth objective. 
For this reason, to improve the performance of our learner we replace every ReLU activation within the neural network $V$ with a smooth approximation, Softplus, which takes the form $\softplus(s) = \log(1 + \exp(s))$, whereas the occurrences of $\relu{}$ applied to the output of $V$ in Eqs.\  \eqref{eq:ind-loss}, \eqref{eq:dec-loss} and \eqref{eq:dec-loss} are retained.
Additionally, we improve the approximation of Softplus to ReLU at small values over the interval $[0, 1]$ by re-scaling $V$. This means that we modify the indicating condition to require that $V(x) \geq \alpha$ over the target set $A$, for some large $\alpha > 1$. In other words, we modify the indicating loss to 
$
\mathcal{L}_\indicating(V)
= \frac{1}{|D \cap A|}
\sum_{d \in D \cap A}\relu( \alpha - V(d))$.
We remark that while $\softplus$ is used as the activation function in the learning stage, for the verification stage we instead employ $\relu$ activation functions, ensuring soundness of the generated neural supermartingale. We remark that this has no effect on the soundness of our approach, which is ultimately guaranteed by the verifier.

\subsection{Verification of Neural Supermartingales Using SMT Solvers}\label{sec:verification} 

The purpose of the verification stage is to check that the neural supermartingale meets the requirements of Eq.~\eqref{eq:spec} over the entire state space $\mathbb{R}^n$, and if that is determined to be the case to furthermore obtain a sound upper bound on the reachability probability.
We achieve this by constructing a suitable formula in first-order logic, Eq.~\eqref{eq:val-q}, and use SMT solving to decide its validity, or equivalently, to decide the satisfiability of its negation, Eq.~\eqref{eq:smt-spec}. 

The conditions pertaining to the validity of the neural supermartingale, given in \eqref{eq:spec_ind} and \eqref{eq:spec-dec}, are encoded by the formulas 
$\varphi_\indicating$ 
and $\varphi_\nonincreasing$. We also note that for a constant $p \in [0, 1]$ to be a sound upper bound on the reachability probability, it is sufficient to require that $p$ is an upper bound on the neural supermartingale's value over the set of initial states $S_0$ (cf.\ \cref{thm:V-of-s}), which is expressed by the formula $\varphi_\verifbound$. A suitable choice for the bound $p$ is determined by a binary search over the interval $[0, 1]$.
\begin{equation}
\label{eq:val-q}
    \forall s \in \Real^n: \underbrace{(s \in A \rightarrow V(s) \geq 1 )}_{\varphi_\indicating} \wedge \underbrace{(s \notin A \rightarrow \mathbb{X}[V](s) \leq V(s))}_{\varphi_\nonincreasing} 
    \wedge \underbrace{(s \in {S}_0 \rightarrow V(s) < p)}_{\varphi_\verifbound}.
\end{equation}
Here, $V(s)$ is a symbolic encoding of the candidate neural supermartingale proposed by the learner (\cref{sec:learner}), and $S_0$ and $A$ are defined by Boolean predicates over program variables, all of which (in our setting of networks composed from ReLU activations) are expressible using expressions and constraints in non-linear real arithmetic. For this reason, we use Z3 as our SMT solver \cite{z3-solver}.

The SMT solver is provided with the negation of Eq.~\eqref{eq:val-q}, namely 
\begin{equation}
    \label{eq:smt-spec}
        \exists s \in \Real^n: \underbrace{(s \in A \wedge V(s) < 1 )}_{\lnot \varphi_\indicating} \vee \underbrace{(s \notin A \wedge \mathbb{X}[V](s) > V(s))}_{\lnot \varphi_\nonincreasing} 
        \vee \underbrace{(s \in {S}_0 \wedge V(s) \geq p)}_{\lnot\varphi_\verifbound},
\end{equation}
and decides its satisfiability, seeking an assignment $d_\text{cex}$ of $s$ that is a counterexample to the neural supermartingale's validity, for which any of $\lnot\varphi_\indicating$, $\lnot\varphi_\nonincreasing$ and $\lnot\varphi_\verifbound$ are satisfied. If no counterexample is found, this certifies the validity of the neural supermartingale. 
Alternatively, if a counterexample $d_\text{cex}$ is found, it is added to the data set $D$, for the synthesis to incrementally resume.

\section{Experimental Evaluation}\label{sec:experiments}

\newcommand{\bmname}[1]{{\ttfamily #1}}
\begin{table}[t]
    \caption{Results comparing neural supermartingales with Farkas Lemma for different benchmarks. Here, $p$ is the average probability bound generated by the certificate; success ratio is the number of successful experiments, out of 10 repeats, generated by CEGIS with neural supermartingale; `-' means no result was obtained. We also denote the architecture of the network: $(h_1, h_2)$ denotes a network with 2 hidden layers consisting of $h_1$ and $h_2$ neurons respectively. }
    \label{tab:results}
    \centering
    \scalebox{0.86}{\begin{tabular}{lcccccc}\toprule
        \textbf{Benchmark}&  \textbf{Farkas'} & \multicolumn{4}{c}{\textbf{Quantitative Neural Certificates}} & \textbf{Network} \\ \cmidrule(lr){3-6} 
        & \textbf{Lemma} & \multicolumn{2}{c}{Program-Agnostic} & \multicolumn{2}{c}{Program-Aware} &
        \textbf{Arch.}\\  \cmidrule(lr){3-4} \cmidrule(lr){5-6} 
        & & $p$ & Success Ratio &  $p$ & Success Ratio \\ 
        \midrule
        \bmname{persist\_2d}        & -             & $\leq 0.1026$ & 0.9 & $\leq 0.1175$ & 0.9  & (3, 1) \\
        \bmname{faulty\_marbles}    & -             & $\leq 0.0739$ & 0.9 & $\leq 0.0649$ & 0.8 & 3 \\
        \bmname{faulty\_unreliable} & -             & $\leq 0.0553$ & 0.9 & $\leq 0.0536$ & 1.0 & 3 \\
        \bmname{faulty\_regions}    & -             & $\leq 0.0473$ & 0.9 & $\leq 0.0411$ & 0.9 & (3, 1) \\\midrule
        \bmname{cliff\_crossing}    & $\leq 0.4546$ & $\leq 0.0553$ & 0.9 & $\leq 0.0591$ & 0.8 & 4 \\
        \bmname{repulse100}            & $\leq 0.0991$ & $\leq 0.0288$ & 1.0 & $\leq 0.0268$ & 1.0 & 3 \\
        \bmname{repulse100\_uniform}   & $\leq 0.0991$ & $\leq 0.0344$ & 1.0 & -             & -   & 2 \\
        \bmname{repulse100\_2d}        & $\leq 0.0991$ & $\leq 0.0568$ & 1.0 & $\leq 0.0541$ & 1.0 & 3 \\
        \bmname{faulty\_varying}    & $\leq 0.1819$ & $\leq 0.0864$ & 1.0 & $\leq 0.0865$ & 1.0 & 2 \\
        \bmname{faulty\_concave}    & $\leq 0.1819$ & $\leq 0.1399$ & 1.0 & $\leq 0.1356$ & 0.9 & (3, 1) \\ \midrule
        \bmname{fixed\_loop}        & $\leq 0.0091$ & $\leq 0.0095$ & 1.0 & $\leq 0.0094$ & 1.0 & 1 \\ 
        \bmname{faulty\_loop}       & $\leq 0.0181$ & $\leq 0.0195$ & 1.0 & $\leq 0.0184$ & 1.0 & 1 \\ 
        \bmname{faulty\_uniform}    & $\leq 0.0181$ & $\leq 0.0233$ & 1.0 & $\leq 0.0221$ & 1.0 & 1 \\ 
        \bmname{faulty\_rare}       & $\leq 0.0019$ & $\leq 0.0022$ & 1.0 & $\leq 0.0022$ & 1.0 & 1 \\ 
        \bmname{faulty\_easy1}      & $\leq 0.0801$ & $\leq 0.1007$ & 1.0 & $\leq 0.0865$ & 1.0 & 1 \\ 
        \bmname{faulty\_ndecr}      & $\leq 0.0561$ & $\leq 0.0723$ & 1.0 & $\leq 0.0630$ & 1.0 & 1 \\
        \bmname{faulty\_walk}       & $\leq 0.0121$ & $\leq 0.0173$ & 1.0 & $\leq 0.0166$ & 1.0 & 1 \\ 
        \bottomrule
    \end{tabular}}
\end{table}

\begin{figure}
    \centering
\begin{tikzpicture}

\definecolor{crimson2143940}{RGB}{214,39,40}
\definecolor{darkgray176}{RGB}{176,176,176}
\definecolor{lightgray204}{RGB}{204,204,204}

\begin{axis}[
legend cell align={left},
legend style={
  fill opacity=0.8,
  draw opacity=1,
  text opacity=1,
  at={(0.98,0.02)},
  anchor=south east,
  draw=lightgray204
},
log basis x={10},
log basis y={10},
tick align=outside,
tick pos=left,
x grid style={darkgray176},
xlabel style={align=center},
xlabel={Probability bound $p$\\Program-Agnostic Neural Supermartingale},
xmin=0.001, xmax=1.2,
xmode=log,
xtick style={color=black},
y grid style={darkgray176},
ylabel style={align=center},
ylabel={Probability bound $p$\\Farkas' Lemma},
ymin=0.001, ymax=1.2,
ymode=log,
ytick style={color=black},
width=2.7in,
height=2.7in
]
\addplot [semithick, crimson2143940, mark=*, mark size=3, mark options={solid}, only marks]
table {%
0.0842 1
0.0739 1
0.0553 1
0.0313 1
};
\addlegendentry{Farkas Fails}
\addplot [semithick, crimson2143940, mark=x, mark size=3, mark options={solid}, only marks]
table {%
0.0288 0.0991
0.0344 0.0991
0.0568 0.0991
0.0864 0.1819
0.1399 0.1819
0.0095 0.0091
0.0195 0.0181
0.0233 0.0181
0.0022 0.0019
0.1007 0.0801
0.0723 0.0561
0.0173 0.0121
0.0553 0.4546
};
\addlegendentry{Farkas Succeeds}
\addplot [semithick, black, dashed]
table {%
0.0001 0.0001
1 1
};
\end{axis}

\end{tikzpicture}
\caption{Probability bounds generated using program-agnostic neural supermartingales and using Farkas' Lemma. The reference line $y=x$ allows one to see which approach outperforms the other and by how much: above the line means that neural supermartingales outperform Farkas' Lemma, below the line the opposite. Neural supermartingales can significantly outperform linear templates when a better bound exists, but otherwise achieve similar results. Our approach with program-aware neural supermartingales provides even better outcomes, compared with Farkas' Lemma. 
    }
    \label{fig:farkas_vs_ism}
\end{figure}
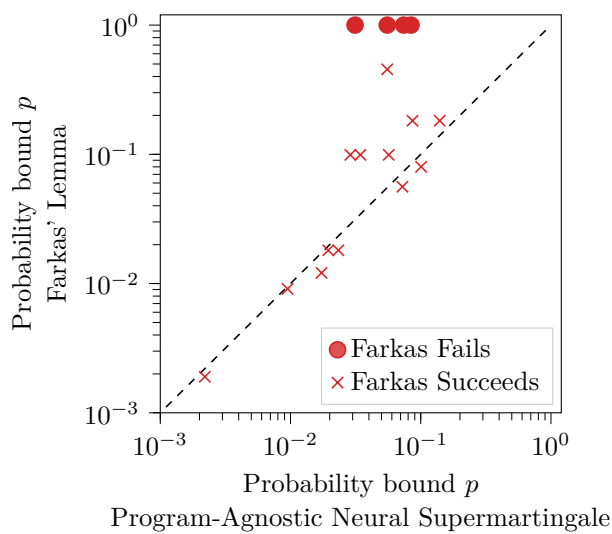

The previous section has presented a new method for synthesising neural supermartingales. This section presents an empirical evaluation of the method, by testing it against a series of benchmarks. Each benchmark is run ten times. A test is successful if a valid supermartingale is synthesised, and the proportion of successful tests is recorded. Further, the average bound from the valid supermartingales is also recorded, along with the average time taken by the learning and verification steps, respectively. This procedure is applied separately for program-aware and program-agnostic synthesis. 
We use values $\beta_1 = 10, \beta_2 = 10^5, \beta_3 = 1$ in Eq.~\eqref{eq:loss} to define the priority ordering of the three terms in the loss function, which we find beneficial across our set of benchmarks. 
To compare our method against existing work, we perform template-based synthesis of linear supermartingales using Farkas' Lemma. This requires deterministic invariants to overapproximate the reachable set of states, which may either be generated by abstract interpretation, or provided manually \cite{DBLP:journals/toplas/TakisakaOUH21}. In our experiments involving the Farkas' Lemma, we provide a suitable invariant manually based on the guard of the loop, in some cases strengthening them with additional constraints by an educated guess. 

It should be noted that our method is inherently stochastic. One reason is the random initialisation of the neural template's parameters in the learning phase. In program-agnostic synthesis, an additional source of randomness is the sampling of successor states. So that the results accurately reflect the performance of our method, the random seed for these sources of randomness is selected differently for each test. An additional source of non-determinism arises from the SMT solver Z3 as it generates counterexamples: this cannot be controlled externally. Benchmarks are run on a machine with an Nvidia A40 GPU, and involve the assertion-violation analysis of programs created using the following two patterns:

\begin{description}
\item[Unreliable Hardware] These are programs that execute on  unreliable hardware. The goal is to upper bound the probability that the program fails to terminate due to a hardware fault. A simple example is \texttt{faulty\_loop}, whose source code is presented in \Cref{lst:faulty}, and which consists of a loop which may violate an assertion with small probability, modelling a hardware fault.
\item[Robot Motion] These programs model an agent (e.g., a robot) that moves within a physical environment. In these benchmarks, the uncertainty in control and sensing is modelled probabilistically. The environment contains a target region and a hazardous region. The goal is to upper bound the probability that the robot enters the hazardous region. We provide the program \texttt{repulse100} as an example (\cref{lst:repulse}), which is a variant of {\tt repulse} (\cref{fig:Intuition}) containing a modified assertion and initial state. The program models the motion of a robot in a one-dimensional environment, starting at ${\tt x} = 10$. The target region is where ${\tt x} < 0$, and hazardous region is where ${\tt x} > 100$. As with {\tt repulse}, in each iteration there is an equal probability of ${\tt x}$ being decremented by 2, and $x$ being incremented by 1.
\end{description}
Several of the programs are based on benchmarks used in prior work focusing on other types of supermartingales \cite{mdranking,DBLP:conf/cav/ChakarovS13}. Additionally, there are several benchmarks that are entirely new.

\begin{table}
    \caption{Results showing the time taken in seconds to synthesise supermartingales by our method and Farkas' Lemma. For our method, we show the time taken during learning and verification.}
    \label{tab:timings-main}
    \centering
    \scalebox{0.88}{\begin{tabular}{lccccc}\toprule
        \textbf{Benchmark}&  \textbf{Farkas'} & \multicolumn{4}{c}{\textbf{Quantitative Neural Certificates}} \\ \cmidrule(lr){3-6} 
        & \textbf{Lemma} & \multicolumn{2}{c}{Program-Agnostic} & \multicolumn{2}{c}{Program-Aware} \\  \cmidrule(lr){3-4} \cmidrule(lr){5-6} 
        & & Learn Time & Verify Time & Learn Time & Verify Time \\ 
        \midrule
        \bmname{persist\_2d}        & -      & 169.14 & 85.31 & 44.96 & 74.90 \\
        \bmname{faulty\_marbles}    & -             & 114.24 &  29.23 & 15.86 &  28.68 \\
        \bmname{faulty\_unreliable} & -             & 123.85 & 45.48 & 18.34 &  33.97 \\
        \bmname{faulty\_regions}    & -             & 17.92 & 35.85 & 17.55 & 32.38 \\\midrule
        \bmname{cliff\_crossing}    & 0.11      & 134.61 & 19.02 & 21.27 & 29.07 \\
        \bmname{repulse100}            & 0.19  & 16.65  & 5.00  & 6.49  & 3.74 \\
        \bmname{repulse100\_uniform}   & 0.19  & 21.28  & 14.18 & -     & - \\
        \bmname{repulse100\_2d}        & 0.12 & 122.92 & 64.54 & 15.75 & 47.70 \\
        \bmname{faulty\_varying}    & 0.36  & 21.74  & 5.06  & 4.71  & 3.28  \\
        \bmname{faulty\_concave}    & 0.39  & 49.12  & 13.37 & 13.49 & 7.82  \\ \midrule
        \bmname{fixed\_loop}        & 0.15  & 14.16  & 3.14  & 3.34  & 2.43  \\ 
        \bmname{faulty\_loop}       & 0.16  & 25.52  & 3.81  & 3.73  & 2.66  \\ 
        \bmname{faulty\_uniform}    & 0.34  & 20.20  & 1.91  & 6.75  & 1.33  \\ 
        \bmname{faulty\_rare}       & 0.27  & 25.52  & 4.27  & 3.71  & 2.96  \\ 
        \bmname{faulty\_easy1}      & 0.31  & 104.20 & 12.78 & 4.95  & 7.51  \\ 
        \bmname{faulty\_ndecr}      & 0.33  & 104.89 & 9.06  & 5.37  & 4.66  \\
        \bmname{faulty\_walk}       & 0.32  & 15.08  & 4.00  & 6.97  & 3.33  \\ 
        \bottomrule
    \end{tabular}}
\end{table}

The results are reported in Table~\ref{tab:results}. The table is divided into three sections. The first section shows the benchmarks where Farkas' Lemma cannot be applied, and where only our method is capable of producing a bound. The second section shows examples where both methods are able to produce a bound, but our method produces a notably better bound. The third section shows benchmarks where both methods produce comparable bounds.

Dashes in the table indicate experiments where a valid supermartingale could not be obtained. In the case of Farkas' Lemma, there are several cases where there is no linear supermartingale for the benchmark. By contrast, program-agnostic and program-aware synthesis could be applied to almost all benchmarks. The one exception is \texttt{repulse100\_uniform} where only program-aware synthesis was unsuccessful: this is due to indicator functions in the post-expectation, which are not smooth and posed a problem for the optimiser. This benchmark underscores the value of program-agnostic synthesis, since it does not require embedding the explicit post-expectation in the loss function. 

The first section of the benchmarks in \cref{tab:results} demonstrates that our method produces useful results on programs that are out-of-scope for existing techniques. Furthermore, the success ratio of our method is high on all the benchmarks, which indicates its robustness. For the second and third sections (which consist of benchmarks to which Farkas' Lemma is applicable), the success ratio of our method is broadly maximal, which is to be expected, since these programs can also be solved by Farkas' Lemma. 

In the second section of \cref{tab:results}, we find more complex benchmarks where our method was able to significantly improve the bound from Farkas' Lemma. The smallest improvement was about 0.04, and the largest improvement was over 0.39. The intuition here is that neural templates allow more sophisticated supermartingales to be learnt, that can approximate how the reachability probability varies across the state space better than linear templates, and thereby yield tighter probability bounds. 

The third section of \cref{tab:results} consists of relatively simple benchmarks, where our method produces results that are marginally less tight in comparison to Farkas' Lemma. This is not surprising since our method uses neural networks consisting of a single neuron for these examples, owing to their simplicity. The expressive power of these networks is therefore similar to linear templates. 

In summary, the results show that our method does significantly better on more complex examples, and marginally worse on very simple examples.
This is highlighted in \Cref{fig:farkas_vs_ism}. Each point represents a benchmark. The position on the $x$-axis shows the probability bound obtained by our program-agnostic method, and the $y$-axis shows the probability bound obtained by Farkas' Lemma. Points above the line indicate benchmarks where neural supermartingales outperform Farkas' Lemma, and vice versa. The scale is logarithmic to emphasise order-of-magnitude differences.

Notice that in \cref{tab:results} the program-aware algorithm usually yields better bounds than the program-agnostic algorithm, but the improvement is mostly marginal. This is in fact a strength of our method: our data-driven approach performs almost as well as one dependent on symbolic representations, which is promising in light of questions of scalability to more complex programs. We also include a breakdown of computation time (\Cref{tab:timings-main}) 
which allows distinguishing between learning and verification overheads. Notably, Farkas' Lemma is significantly faster than our method, given that it relies on solving a convex optimisation problem via linear programming, whereas the synthesis of neural supermartingales is a non-convex optimisation problem that is addressed using gradient descent.

\begin{figure}[t!]
\begin{minipage}[t]{0.48\linewidth}
\begin{lstlisting}[caption={{\tt faulty\_loop}}, label={lst:faulty}]
error = 0;
i = 0;
while i < 10 do
    assert(error == 0);
    p ~ Bernoulli(0.999);
    if p == 1 then
        q ~ Bernoulli(0.5);
        if q == 1 then
            i += 1
        fi
    else
        error = 1
    fi
od
\end{lstlisting}
\end{minipage}
\begin{minipage}[t]{0.48\linewidth}
\begin{lstlisting}[caption={{\tt repulse100} }, label={lst:repulse}]
x = 10;
while x >= 0 do
    assert(x <= 100);
    p ~ Bernoulli(0.5);
    if p == 1 then
        x -= 2
    else
        x += 1
    fi
od
\end{lstlisting}
\end{minipage}

\begin{minipage}[t]{0.48\linewidth}
\begin{lstlisting}[caption={{\tt cliff\_crossing} },label={lst:cliff}]
x = 0;
y = 0;
while x <= 50 do
    assert(y <= 10);
    p ~ Bernoulli(0.9);
    x += 1;
    if p == 0 then
        y += 1
    else
        if y >= 1 then
            y -= 1
        fi
    fi
od
\end{lstlisting}
\end{minipage}
\begin{minipage}[t]{0.48\linewidth}
\begin{lstlisting}[caption={{\tt faulty\_concave} },label={lst:concave}]
i = 1;
error = 0;
while i < 10 do
    assert(error == 0);
    if i < 4 then 
        p ~ Bernoulli(0.999)
    else
        p ~ Bernoulli(0.99)
    fi;
    if p == 1 then
        q ~ Bernoulli(0.5);
        if q == 1 then
            i += 1
        fi
    else
        error = 1
    fi
od
\end{lstlisting}
\end{minipage}
\end{figure}

Having presented the experimental results, we shall further comment on some specific benchmarks.
The {\tt repulse100} program (\cref{lst:repulse}) is a variation of {\tt repulse} presented in \cref{fig:Intuition}, but with the assertion changed to {\tt assert(x <= 100)}. While this is a small program, our method is still able to produce a significantly better result than Farkas' lemma, using a neural supermartingale with a single hidden layer consisting of three ReLU components that are summed together, which allows a convex piecewise linear function to be learnt. 

The {\tt cliff\_crossing} program (\cref{lst:cliff}) is a further benchmark for which our method is capable of producing a significantly better bound. 
\texttt{cliff\_crossing} (\cref{lst:cliff}) models a robot moving along a road (defined by the region $(x,y) \in [0, 50] \times [0, 10]$) that is adjacent to a cliff (defined by the region with $y > 10$). 
This is a 2 dimensional benchmark, for which we use a neural supermartingale that consists of two input neurons and four ReLU components, leading to a clear improvement compared to the linear supermartingale in the tightness of the probability bounds generated, as illustrated by \cref{fig:cliff-main}.

Both {\tt repulse100} and {\tt cliff\_crossing} are benchmarks that use neural supermartingales with a single hidden layer. An example that uses two hidden layers is {\tt faulty\_concave} (\cref{lst:concave}),
which simulates a program executing on unreliable hardware, where the probability of a fault varies across the state space.
This results in two distinct regions of the state space, one of which has a significantly higher probability of assertion violation than the other. We find that neither a linear template nor a single-layer neural supermartingale is able to exploit this conditional behaviour, each of which yield an overly conservative certificate, but that a neural supermartingale with two hidden layers is able to more tightly approximate the reachability probability in each of the two regions.

\begin{figure}[t]
     \begin{subfigure}[b]{0.49\textwidth}
        \centering
        \begin{tikzpicture}
        \begin{axis}[mesh/ordering=y varies,colormap/greenyellow,width=2.9in]
        	\addplot3 [surf,mesh/rows=11] file {graphs/cliff_crossing_linear.dat};
        \end{axis}
        \end{tikzpicture}
        \caption*{Linear supermartingale.}
        \label{fig:cliff-lin-main}
    \end{subfigure}
    \hfill
    \begin{subfigure}[b]{0.49\textwidth}
        \centering
        \begin{tikzpicture}
        \begin{axis}[mesh/ordering=y varies,colormap/cool,width=2.9in]
	    \addplot3 [surf,mesh/rows=11] file {graphs/cliff_crossing_neural.dat};
        \end{axis}
        \end{tikzpicture}
        \caption*{Neural supermartingale.}
        \label{fig:cliff-nueral-main}
    \end{subfigure}
    \caption{Supermartingale functions for the \texttt{cliff\_crossing} benchmark as generated using Farkas' Lemma (on the left) and using neural supermartingales (on the right). The right hand figure illustrates the tighter bounds obtainable through the use of neural templates.
    }
    \label{fig:cliff-main}
\end{figure}
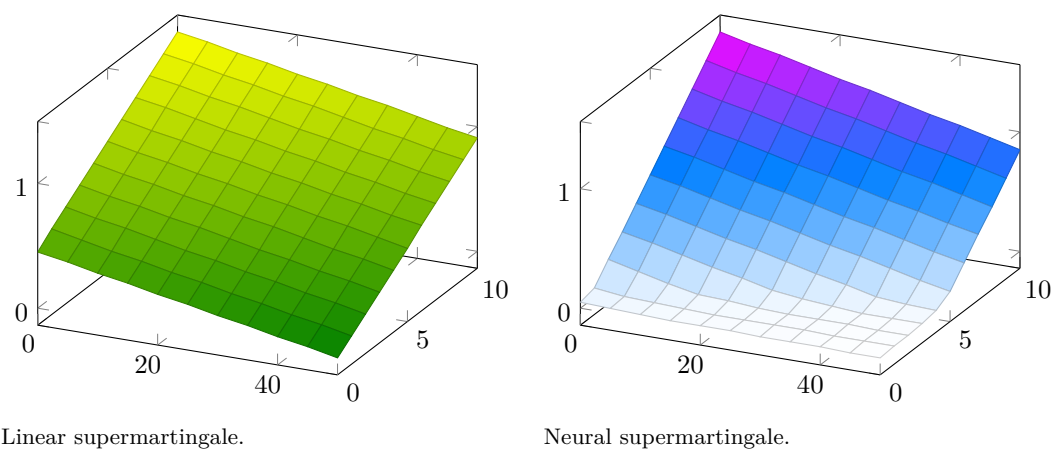

\section{Related Work}

The formal verification of probabilistic programs using supermartingales is a well-studied topic. Early approaches to introduce this technique applied them to almost-sure termination analysis of probabilistic programs \cite{DBLP:conf/cav/ChakarovS13}, which allowed several extensions to polynomial programs, programs with non-determinism,
lexicographic and modular termination arguments, and
persistence properties~\cite{AgrawalC018,DBLP:conf/cav/ChatterjeeFG16,DBLP:journals/toplas/ChatterjeeFNH18,DBLP:conf/popl/ChatterjeeNZ17,FuC19,Huang0CG19}. All these methods relied on symbolic reasoning algorithms for synthesising supermartingales, that leveraged theories based on Farkas' lemma for the synthesis of linear certificates, and Putinar's Positivestellensatz and sum-of-square methods for the synthesis of polynomial certificates~\cite{DBLP:conf/atva/ChatterjeeGGKSSZ25}. While these methods are the state-of-the-art for many existing problem instances in literature, to achieve the strong guarantees that they provide (such as completeness for the specific class of programs they target), they must necessarily introduce restrictions on the class of programs to which they are applicable, and the form of certificates that they derive. 
Moreover, symbolic methods need externally provided invariants that are stronger than $\mathbb{R}^n$ to enforce non-negativity in the case of linear certificates, as we illustrate in \cref{fig:Intuition}. Also, symbolic methods for the synthesis of polynomial certificates require compact deterministic invariants to operate. 

The use of neural networks to represent certificates has allowed many of these restrictions to be lifted. In the context of the analysis of probabilistic programs, neural networks were first applied to certify positive almost-sure termination \cite{DBLP:conf/cav/AbateGR20}. This approach lent itself to a wider range of formal verification questions for stochastic dynamical models, from stability and safety analysis to controller synthesis~\cite{LSAZ22,DBLP:conf/aaai/LechnerZCH22,DBLP:conf/tacas/ChatterjeeHLZ23,DBLP:journals/csysl/MathiesenCL23,DBLP:conf/aaai/NadaliM0024,DBLP:conf/aaai/NeustroevGL25}. These data-driven inductive synthesis techniques for supermartingales have also been extended to machine learning techniques other then deep learning, such as piecewise linear regression and decision tree learning~\cite{DBLP:conf/tacas/BatzCJKKM23,DBLP:conf/cav/BaoTPHR22}. Data-driven techniques for formal verification have also been applied to the construction of system abstractions, reducing the verification of the original system to a simpler yet equivalent verification task~\cite{DBLP:conf/cav/AbateGS24,DBLP:conf/nips/AbateEG22,DBLP:conf/cav/AbateGMS25,DBLP:journals/corr/abs-2505-15497}. 

The present work focuses on verifying the system under analysis through the construction of proof certificates. In particular, our work extends the data-driven synthesis of neural supermartingale certificates to quantitative verification questions. The correctness of our approach builds upon the theory of non-negative repulsing supermartingales \cite{DBLP:journals/toplas/TakisakaOUH21}, as formulated in \cref{thm:V-of-s}. 
Our experiments have demonstrated that neural certificates attain comparable results on programs that are amenable to symbolic analysis (such as those in the third section of \cref{tab:results}), while surpassing symbolic methods on more complex programs that are either out-of-scope or yield overly conservative bounds when existing techniques are applied (such as those in the first and second section of \cref{tab:results}).

\section{Conclusion}

We have presented a data-driven framework for the quantitative verification of probabilistic models 
that leverage neural networks to represent supermartingale certificates. 
Our experiments have shown that neural certificates are applicable to a wider range of probabilistic models than was previously possible using purely symbolic techniques. We also illustrate that on existing models our method yields certificates of better or comparable quality than those produced by symbolic techniques for the synthesis of linear supermartingales. This builds upon the ability of neural networks to approximate non-linear functions, while satisfying the constraints imposed by \cref{thm:V-of-s} without the need for supporting deterministic invariants to be provided externally. Our method applies to quantitative termination and assertion-violation analysis for probabilistic programs, as well as safety and invariant verification for stochastic dynamic models. We imagine extensions to further quantitative verification questions, such as temporal properties beyond reachability \cite{DBLP:books/daglib/0020348,DBLP:conf/cav/AbateGR24,neuralmc,DBLP:conf/cav/AbateGR25,DBLP:conf/nips/GiacobbeKPT25}, and bounding expected accrued costs \cite{ngoBoundedExpectationsResource2018, wangCostAnalysisNondeterministic2019b, wangCentralMomentAnalysis2021}. While neural networks of no greater than two layers suffice for our examples, it also remains open to investigate whether deeper network architectures are beneficial for quantitative verification, as well as their effect on the tightness of bounds obtained and scalability of the approach.

\bibliographystyle{alphaurl} 
\bibliography{main}

@article{c,
    title = "Probabilistic reachability and safety for controlled discrete time stochastic hybrid systems",
    journal = "Automatica",
    volume = "44",
    number = "11",
    pages = "2724 - 2734",
    year = "2008",
    author = "Alessandro Abate and Maria Prandini and John Lygeros and Shankar Sastry"
}

@book{revuz1975markov,
  title={Markov Chains},
  author={Revuz, D.},
  isbn={9780444107527},
  lccn={74080112},
  series={Mathematical Studies},
  year={1975},
  publisher={North-Holland Publishing Company}
}

@book{pollard_2001, 
    place={Cambridge}, 
    series={Cambridge Series in Statistical and Probabilistic Mathematics}, title={A User's Guide to Measure Theoretic Probability}, publisher={Cambridge University Press}, author={Pollard, David}, year={2001}, collection={Cambridge Series in Statistical and Probabilistic Mathematics}}

@book{meyn_tweedie_glynn_2009, 
    place={Cambridge}, 
    edition={2}, 
    series={Cambridge Mathematical Library}, title={Markov Chains and Stochastic Stability}, 
    publisher={Cambridge University Press},
    author={Meyn, Sean and Tweedie, Richard L. and Glynn, Peter W.}, 
    year={2009}, collection={Cambridge Mathematical Library}
}

@book{Durrett2010,
  author       = {Rick Durrett},
  title        = {Probability: Theory and Examples, 4th Edition},
  publisher    = {Cambridge University Press},
  year         = {2010}
}

@inproceedings{DBLP:conf/popl/ChatterjeeNZ17,
  author    = {Krishnendu Chatterjee and
               Petr Novotn{\'{y}} and
               \DJ{}or\dj{}e \v{Z}ikeli\'{c}},
  title     = {Stochastic invariants for probabilistic termination},
  booktitle = {{POPL}},
  pages     = {145--160},
  publisher = {{ACM}},
  year      = {2017}
}

@book{davey_priestley_2002, place={Cambridge}, edition={2}, title={Introduction to Lattices and Order}, DOI={10.1017/CBO9780511809088}, publisher={Cambridge University Press}, author={Davey, B. A. and Priestley, H. A.}, year={2002}}

@book{royden2010real,
  title={Real Analysis},
  author={Royden, H.L. and Fitzpatrick, P.},
  isbn={9780131437470},
  lccn={2009048692},
  year={2010},
  edition={4},
  publisher={Prentice Hall}
}

@book{axler2019measure,
  title={Measure, Integration \& Real Analysis},
  author={Axler, S.},
  isbn={9783030331436},
  series={Graduate Texts in Mathematics},
  year={2019},
  publisher={Springer International Publishing}
}

@article{Huang0CG19,
  author    = {Mingzhang Huang and
               Hongfei Fu and
               Krishnendu Chatterjee and
               Amir Kafshdar Goharshady},
  title     = {Modular verification for almost-sure termination of probabilistic
               programs},
  journal   = {Proc. {ACM} Program. Lang.},
  volume    = {3},
  number    = {{OOPSLA}},
  pages     = {129:1--129:29},
  year      = {2019}
}

@inproceedings{FuC19,
  author    = {Hongfei Fu and
               Krishnendu Chatterjee},
  title     = {Termination of Nondeterministic Probabilistic Programs},
  booktitle = {{VMCAI}},
  series    = {LNCS},
  volume    = {11388},
  pages     = {468--490},
  publisher = {Springer},
  year      = {2019}
}

@book{BS96,
	author = {Bertsekas, D. P. and Shreve , S. E.},
	publisher = {Athena Scientific},
	title = {Stochastic optimal control: {T}he discrete-time case},
	year = {1996}}

@book{hll1996,
	author = {Hern{\'a}ndez-Lerma, O. and Lasserre, J. B.},
	pages = {xiv+216},
	publisher = {Springer-Verlag},
	series = {Applications of Mathematics},
	title = {Discrete-time {M}arkov control processes},
	volume = {30},
	year = {1996}}

@book{kallenberg,
	author = {Kallenberg, O.},
	publisher = {Springer Science \& Business Media},
	title = {Foundations of modern probability},
	year = {2006}}

@inproceedings{DBLP:conf/asplos/Solar-LezamaTBSS06,
  author    = {Armando Solar{-}Lezama and
               Liviu Tancau and
               Rastislav Bod{\'{\i}}k and
               Sanjit A. Seshia and
               Vijay A. Saraswat},
  title     = {Combinatorial sketching for finite programs},
  booktitle = {{ASPLOS}},
  pages     = {404--415},
  publisher = {{ACM}},
  year      = {2006}
}

@phdthesis{solar-lezama2008ProgramSynthesisSketching,
  title = {Program Synthesis by Sketching},
  author = {{Solar-Lezama}, Armando},
  year = {2008},
  address = {{USA}},
  advisor = {Bodik, Rastislav},
  isbn = {9781109097450},
  school = {University of California at Berkeley}
}

@inproceedings{z3-solver,
  author       = {Leonardo Mendon{\c{c}}a de Moura and
                  Nikolaj S. Bj{\o}rner},
  title        = {{Z3:} An Efficient {SMT} Solver},
  booktitle    = {{TACAS}},
  series       = {Lecture Notes in Computer Science},
  volume       = {4963},
  pages        = {337--340},
  publisher    = {Springer},
  year         = {2008}
}

@inproceedings{gehr2016psi,
  author       = {Timon Gehr and
                  Sasa Misailovic and
                  Martin T. Vechev},
  title        = {{PSI:} Exact Symbolic Inference for Probabilistic Programs},
  booktitle    = {{CAV} {(1)}},
  series       = {Lecture Notes in Computer Science},
  volume       = {9779},
  pages        = {62--83},
  publisher    = {Springer},
  year         = {2016}
}

@inproceedings{gehr2020lpsi,
  author       = {Timon Gehr and
                  Samuel Steffen and
                  Martin T. Vechev},
  title        = {{\(\lambda\)PSI}: exact inference for higher-order probabilistic programs},
  booktitle    = {{PLDI}},
  pages        = {883--897},
  publisher    = {{ACM}},
  year         = {2020}
}

@book{DBLP:books/daglib/0020348,
  author    = {Christel Baier and
               Joost{-}Pieter Katoen},
  title     = {Principles of model checking},
  publisher = {{MIT} Press},
  year      = {2008}
}

@article{DBLP:journals/toplas/TakisakaOUH21,
  author    = {Toru Takisaka and
               Yuichiro Oyabu and
               Natsuki Urabe and
               Ichiro Hasuo},
  title     = {Ranking and Repulsing Supermartingales for Reachability in Randomized
               Programs},
  journal   = {{ACM} Trans. Program. Lang. Syst.},
  volume    = {43},
  number    = {2},
  pages     = {5:1--5:46},
  year      = {2021}
}

@inproceedings{DBLP:conf/atva/TakisakaOUH18,
  author    = {Toru Takisaka and
               Yuichiro Oyabu and
               Natsuki Urabe and
               Ichiro Hasuo},
  title     = {Ranking and Repulsing Supermartingales for Reachability in Probabilistic
               Programs},
  booktitle = {{ATVA}},
  series    = {Lecture Notes in Computer Science},
  volume    = {11138},
  pages     = {476--493},
  publisher = {Springer},
  year      = {2018}
}

@inproceedings{DBLP:conf/cav/ChatterjeeGMZ22,
  author    = {Krishnendu Chatterjee and
               Amir Kafshdar Goharshady and
               Tobias Meggendorfer and
               \DJ{}or\dj{}e \v{Z}ikeli\'{c}},
  title     = {Sound and Complete Certificates for Quantitative Termination Analysis
               of Probabilistic Programs},
  booktitle = {{CAV} {(1)}},
  series    = {Lecture Notes in Computer Science},
  volume    = {13371},
  pages     = {55--78},
  publisher = {Springer},
  year      = {2022}
}

@book{DBLP:books/daglib/0070910,
  author       = {Glynn Winskel},
  title        = {The formal semantics of programming languages - an introduction},
  series       = {Foundation of computing series},
  publisher    = {{MIT} Press},
  year         = {1993}
}

@article{DBLP:journals/csysl/AbateAGP21,
  author    = {Alessandro Abate and
               Daniele Ahmed and
               Mirco Giacobbe and
               Andrea Peruffo},
  title     = {Formal Synthesis of {Lyapunov Neural Networks}},
  journal   = {{IEEE} Control. Syst. Lett.},
  volume    = {5},
  number    = {3},
  pages     = {773--778},
  year      = {2021}
}

@inproceedings{giacobbe2022,
  author    = {Mirco Giacobbe and
               Daniel Kroening and
               Julian Parsert},
  title     = {Neural termination analysis},
  booktitle = {{ESEC/SIGSOFT} {FSE}},
  pages     = {633--645},
  publisher = {{ACM}},
  year      = {2022}
}

@inproceedings{DBLP:conf/aaai/LechnerZCH22,
  author    = {Mathias Lechner and
               \DJ{}or\dj{}e \v{Z}ikeli\'{c} and
               Krishnendu Chatterjee and
               Thomas A. Henzinger},
  title     = {Stability Verification in Stochastic Control Systems via Neural Network
               Supermartingales},
  booktitle = {{AAAI}},
  pages     = {7326--7336},
  publisher = {{AAAI} Press},
  year      = {2022}
}

@inproceedings{DBLP:conf/nips/ChangRG19,
  author    = {Ya{-}Chien Chang and
               Nima Roohi and
               Sicun Gao},
  title     = {{Neural Lyapunov Control}},
  booktitle = {NeurIPS},
  pages     = {3240--3249},
  year      = {2019}
}

@inproceedings{DBLP:conf/cav/BaoTPHR22,
  author       = {Jialu Bao and
                  Nitesh Trivedi and
                  Drashti Pathak and
                  Justin Hsu and
                  Subhajit Roy},
  title        = {Data-Driven Invariant Learning for Probabilistic Programs},
  booktitle    = {{CAV} {(1)}},
  series       = {Lecture Notes in Computer Science},
  volume       = {13371},
  pages        = {33--54},
  publisher    = {Springer},
  year         = {2022}
}

@inproceedings{ngoBoundedExpectationsResource2018,
  author       = {Van Chan Ngo and
                  Quentin Carbonneaux and
                  Jan Hoffmann},
  title        = {Bounded expectations: resource analysis for probabilistic programs},
  booktitle    = {{PLDI}},
  pages        = {496--512},
  publisher    = {{ACM}},
  year         = {2018}
}

@inproceedings{wangCostAnalysisNondeterministic2019b,
  author       = {Peixin Wang and
                  Hongfei Fu and
                  Amir Kafshdar Goharshady and
                  Krishnendu Chatterjee and
                  Xudong Qin and
                  Wenjun Shi},
  title        = {Cost analysis of nondeterministic probabilistic programs},
  booktitle    = {{PLDI}},
  pages        = {204--220},
  publisher    = {{ACM}},
  year         = {2019}
}

@inproceedings{DBLP:journals/jmlr/MaillardMLG10,
  author       = {Odalric{-}Ambrym Maillard and
                  R{\'{e}}mi Munos and
                  Alessandro Lazaric and
                  Mohammad Ghavamzadeh},
  title        = {Finite-sample Analysis of Bellman Residual Minimization},
  booktitle    = {{ACML}},
  series       = {{JMLR} Proceedings},
  volume       = {13},
  pages        = {299--314},
  publisher    = {JMLR.org},
  year         = {2010}
}

@inproceedings{DBLP:conf/ijcai/KearnsMN99,
  author       = {Michael J. Kearns and
                  Yishay Mansour and
                  Andrew Y. Ng},
  title        = {A Sparse Sampling Algorithm for Near-Optimal Planning in Large Markov
                  Decision Processes},
  booktitle    = {{IJCAI}},
  pages        = {1324--1231},
  publisher    = {Morgan Kaufmann},
  year         = {1999}
}

@inproceedings{DBLP:conf/icml/JinS21,
  author       = {Yujia Jin and
                  Aaron Sidford},
  title        = {Towards Tight Bounds on the Sample Complexity of Average-reward MDPs},
  booktitle    = {{ICML}},
  series       = {Proceedings of Machine Learning Research},
  volume       = {139},
  pages        = {5055--5064},
  publisher    = {{PMLR}},
  year         = {2021}
}

@book{asmussen2007stochastic,
  title={Stochastic Simulation: Algorithms and Analysis},
  author={Asmussen, S. and Glynn, P.W.},
  isbn={9780387690339},
  lccn={2007926471},
  series={Stochastic Modelling and Applied Probability},
  url={https://books.google.com.sg/books?id=vMI2MdxchU0C},
  year={2007},
  publisher={Springer New York}
}

@inproceedings{wangCentralMomentAnalysis2021,
  author       = {Di Wang and
                  Jan Hoffmann and
                  Thomas W. Reps},
  title        = {Central moment analysis for cost accumulators in probabilistic programs},
  booktitle    = {{PLDI}},
  pages        = {559--573},
  publisher    = {{ACM}},
  year         = {2021}
}

@inproceedings{DBLP:conf/cav/ChakarovS13,
  author       = {Aleksandar Chakarov and
                  Sriram Sankaranarayanan},
  title        = {Probabilistic Program Analysis with Martingales},
  booktitle    = {{CAV}},
  series       = {Lecture Notes in Computer Science},
  volume       = {8044},
  pages        = {511--526},
  publisher    = {Springer},
  year         = {2013}
}

@inproceedings{DBLP:conf/cav/ChatterjeeFG16,
  author       = {Krishnendu Chatterjee and
                  Hongfei Fu and
                  Amir Kafshdar Goharshady},
  title        = {Termination Analysis of Probabilistic Programs Through Positivstellensatz's},
  booktitle    = {{CAV} {(1)}},
  series       = {Lecture Notes in Computer Science},
  volume       = {9779},
  pages        = {3--22},
  publisher    = {Springer},
  year         = {2016}
}

@inproceedings{DBLP:conf/cav/AbateGR20,
  author       = {Alessandro Abate and
                  Mirco Giacobbe and
                  Diptarko Roy},
  title        = {Learning Probabilistic Termination Proofs},
  booktitle    = {{CAV} {(2)}},
  series       = {Lecture Notes in Computer Science},
  volume       = {12760},
  pages        = {3--26},
  publisher    = {Springer},
  year         = {2021}
}

@inproceedings{DBLP:conf/concur/AbateEGPR23,
  author       = {Alessandro Abate and
                  Alec Edwards and
                  Mirco Giacobbe and
                  Hashan Punchihewa and
                  Diptarko Roy},
  title        = {Quantitative Verification with Neural Networks},
  booktitle    = {{CONCUR}},
  series       = {LIPIcs},
  volume       = {279},
  pages        = {22:1--22:18},
  publisher    = {Schloss Dagstuhl - Leibniz-Zentrum f{\"{u}}r Informatik},
  year         = {2023}
}

@incollection{semantics,
author={Dahlqvist, Fredrik and Silva, Alexandra},
title={Semantics of Probabilistic Programming: A Gentle Introduction},
booktitle={Foundations of Probabilistic Programming},
editor={Barthe, Gilles and Katoen, Joost-Pieter and Silva, Alexandra},
publisher={Cambridge University Press},
year={2020},
pages={1-42}}

@inproceedings{GordonHNR14,
  author    = {Andrew D. Gordon and
               Thomas A. Henzinger and
               Aditya V. Nori and
               Sriram K. Rajamani},
  title     = {Probabilistic programming},
  booktitle = {{FOSE}},
  pages     = {167--181},
  publisher = {{ACM}},
  year      = {2014}
}

@article{Kozen81,
  author    = {Dexter Kozen},
  title     = {Semantics of Probabilistic Programs},
  journal   = {J. Comput. Syst. Sci.},
  volume    = {22},
  number    = {3},
  pages     = {328--350},
  year      = {1981}
}

@book{McIverM05,
  author    = {Annabelle McIver and
               Carroll Morgan},
  title     = {Abstraction, Refinement and Proof for Probabilistic Systems},
  series    = {Monographs in Computer Science},
  publisher = {Springer},
  year      = {2005}
}

@article{LSAZ22,
  author = {A. Lavaei and S. Soudjani and A. Abate and M. Zamani},
  title = {Automated Verification and Synthesis of Stochastic Hybrid Systems: A Survey},
  journal = {Automatica},
  year = {2022},
  volume = {146},
  number = {12},
  pages = {}
}

@article{SA13,
  author = {S. Esmaeil Zadeh Soudjani and A. Abate},
  title = {Adaptive and Sequential Gridding for Abstraction and Verification of Stochastic Processes},
  journal = {SIAM Journal on Applied Dynamical Systems},
  year = {2012},
  volume = {12},
  number = {2},
  pages = {921--956}
}

@inproceedings{DBLP:conf/tacas/PeruffoAA21,
  author       = {Andrea Peruffo and
                  Daniele Ahmed and
                  Alessandro Abate},
  title        = {Automated and Formal Synthesis of Neural Barrier Certificates for
                  Dynamical Models},
  booktitle    = {{TACAS} {(1)}},
  series       = {Lecture Notes in Computer Science},
  volume       = {12651},
  pages        = {370--388},
  publisher    = {Springer},
  year         = {2021}
}

@inproceedings{mdranking,
  author       = {Christophe Alias and
                  Alain Darte and
                  Paul Feautrier and
                  Laure Gonnord},
  title        = {Multi-dimensional Rankings, Program Termination, and Complexity Bounds
                  of Flowchart Programs},
  booktitle    = {{SAS}},
  series       = {Lecture Notes in Computer Science},
  volume       = {6337},
  pages        = {117--133},
  publisher    = {Springer},
  year         = {2010}
}

@inproceedings{DBLP:conf/sigsoft/Kwiatkowska07,
  author       = {Marta Z. Kwiatkowska},
  title        = {Quantitative verification: models techniques and tools},
  booktitle    = {{ESEC/SIGSOFT} {FSE}},
  pages        = {449--458},
  publisher    = {{ACM}},
  year         = {2007}
}

@inproceedings{DBLP:conf/icalp/BreugelW01,
  author       = {Franck van Breugel and
                  James Worrell},
  title        = {Towards Quantitative Verification of Probabilistic Transition Systems},
  booktitle    = {{ICALP}},
  series       = {Lecture Notes in Computer Science},
  volume       = {2076},
  pages        = {421--432},
  publisher    = {Springer},
  year         = {2001}
}

@article{DBLP:journals/cacm/BaierHHK10,
  author       = {Christel Baier and
                  Boudewijn R. Haverkort and
                  Holger Hermanns and
                  Joost{-}Pieter Katoen},
  title        = {Performance evaluation and model checking join forces},
  journal      = {Commun. {ACM}},
  volume       = {53},
  number       = {9},
  pages        = {76--85},
  year         = {2010}
}

@inproceedings{DBLP:journals/entcs/KwiatkowskaNP06,
  author       = {Marta Z. Kwiatkowska and
                  Gethin Norman and
                  David Parker},
  title        = {Quantitative Analysis With the Probabilistic Model Checker {PRISM}},
  booktitle    = {{QAPL}},
  series       = {Electronic Notes in Theoretical Computer Science},
  volume       = {153},
  pages        = {5--31},
  publisher    = {Elsevier},
  year         = {2005}
}

@inproceedings{DBLP:conf/tacas/ForejtKNPQ11,
  author       = {Vojtech Forejt and
                  Marta Z. Kwiatkowska and
                  Gethin Norman and
                  David Parker and
                  Hongyang Qu},
  title        = {Quantitative Multi-objective Verification for Probabilistic Systems},
  booktitle    = {{TACAS}},
  series       = {Lecture Notes in Computer Science},
  volume       = {6605},
  pages        = {112--127},
  publisher    = {Springer},
  year         = {2011}
}

@inproceedings{DBLP:conf/hybrid/AbateKM11,
  author       = {Alessandro Abate and
                  Joost{-}Pieter Katoen and
                  Alexandru Mereacre},
  title        = {Quantitative automata model checking of autonomous stochastic hybrid
                  systems},
  booktitle    = {{HSCC}},
  pages        = {83--92},
  publisher    = {{ACM}},
  year         = {2011}
}

@article{DBLP:journals/iandc/TkachevMKA17,
  author       = {Ilya Tkachev and
                  Alexandru Mereacre and
                  Joost{-}Pieter Katoen and
                  Alessandro Abate},
  title        = {Quantitative model-checking of controlled discrete-time Markov processes},
  journal      = {Inf. Comput.},
  volume       = {253},
  pages        = {1--35},
  year         = {2017}
}

@inproceedings{DBLP:conf/lpar/McIverM02,
  author       = {Annabelle McIver and
                  Carroll Morgan},
  title        = {Games, Probability and the Quantitative {\(\mathrm{\mu}\)}-Calculus
                  qM{\(\mathrm{\mu}\)}},
  booktitle    = {{LPAR}},
  series       = {Lecture Notes in Computer Science},
  volume       = {2514},
  pages        = {292--310},
  publisher    = {Springer},
  year         = {2002}
}

@incollection{DBLP:reference/mc/BaierAFK18,
  author       = {Christel Baier and
                  Luca de Alfaro and
                  Vojtech Forejt and
                  Marta Kwiatkowska},
  title        = {Model Checking Probabilistic Systems},
  booktitle    = {Handbook of Model Checking},
  pages        = {963--999},
  publisher    = {Springer},
  year         = {2018}
}

@article{DBLP:journals/toplas/MorganMS96,
  author       = {Carroll Morgan and
                  Annabelle McIver and
                  Karen Seidel},
  title        = {Probabilistic Predicate Transformers},
  journal      = {{ACM} Trans. Program. Lang. Syst.},
  volume       = {18},
  number       = {3},
  pages        = {325--353},
  year         = {1996}
}

@article{DBLP:journals/csysl/MathiesenCL23,
  author       = {Frederik Baymler Mathiesen and
                  Simeon Craig Calvert and
                  Luca Laurenti},
  title        = {Safety Certification for Stochastic Systems via Neural Barrier Functions},
  journal      = {{IEEE} Control. Syst. Lett.},
  volume       = {7},
  pages        = {973--978},
  year         = {2023}
}

@article{DBLP:journals/toplas/ChatterjeeFNH18,
  author       = {Krishnendu Chatterjee and
                  Hongfei Fu and
                  Petr Novotn{\'{y}} and
                  Rouzbeh Hasheminezhad},
  title        = {Algorithmic Analysis of Qualitative and Quantitative Termination Problems
                  for Affine Probabilistic Programs},
  journal      = {{ACM} Trans. Program. Lang. Syst.},
  volume       = {40},
  number       = {2},
  pages        = {7:1--7:45},
  year         = {2018}
}

@article{AgrawalC018,
  author    = {Sheshansh Agrawal and
               Krishnendu Chatterjee and
               Petr Novotn{\'{y}}},
  title     = {Lexicographic ranking supermartingales: an efficient approach to termination
               of probabilistic programs},
  journal   = {Proc. {ACM} Program. Lang.},
  volume    = {2},
  number    = {{POPL}},
  pages     = {34:1--34:32},
  year      = {2018}
}

@inproceedings{DBLP:conf/tacas/ChatterjeeHLZ23,
  author       = {Krishnendu Chatterjee and
                  Thomas A. Henzinger and
                  Mathias Lechner and
                  \DJ{}or\dj{}e \v{Z}ikeli\'{c}},
  title        = {A Learner-Verifier Framework for Neural Network Controllers and Certificates
                  of Stochastic Systems},
  booktitle    = {{TACAS} {(1)}},
  series       = {Lecture Notes in Computer Science},
  volume       = {13993},
  pages        = {3--25},
  publisher    = {Springer},
  year         = {2023}
}

@inproceedings{DBLP:conf/tacas/BatzCJKKM23,
  author       = {Kevin Batz and
                  Mingshuai Chen and
                  Sebastian Junges and
                  Benjamin Lucien Kaminski and
                  Joost{-}Pieter Katoen and
                  Christoph Matheja},
  title        = {Probabilistic Program Verification via Inductive Synthesis of Inductive
                  Invariants},
  booktitle    = {{TACAS} {(2)}},
  series       = {Lecture Notes in Computer Science},
  volume       = {13994},
  pages        = {410--429},
  publisher    = {Springer},
  year         = {2023}
}

@inproceedings{neuralmc,
  author    = {Mirco Giacobbe and
               Daniel Kroening and
               Abhinandan Pal and
	       Michael Tautschnig},
  title     = {Neural Model Checking},
  year      = {2024},
  booktitle = {NeurIPS}
}

@inproceedings{DBLP:conf/cav/AbateGR24,
  author       = {Alessandro Abate and
                  Mirco Giacobbe and
                  Diptarko Roy},
  title        = {Stochastic Omega-Regular Verification and Control with Supermartingales},
  booktitle    = {{CAV} {(3)}},
  series       = {Lecture Notes in Computer Science},
  volume       = {14683},
  pages        = {395--419},
  publisher    = {Springer},
  year         = {2024}
}

@inproceedings{DBLP:conf/aaai/ZikelicLHC23,
  author       = {\DJ{}or\dj{}e \v{Z}ikeli\'{c} and
                  Mathias Lechner and
                  Thomas A. Henzinger and
                  Krishnendu Chatterjee},
  title        = {Learning Control Policies for Stochastic Systems with Reach-Avoid
                  Guarantees},
  booktitle    = {{AAAI}},
  pages        = {11926--11935},
  publisher    = {{AAAI} Press},
  year         = {2023}
}

@inproceedings{DBLP:conf/nips/ZikelicLVCH23,
  author       = {\DJ{}or\dj{}e \v{Z}ikeli\'{c} and
                  Mathias Lechner and
                  Abhinav Verma and
                  Krishnendu Chatterjee and
                  Thomas A. Henzinger},
  title        = {Compositional Policy Learning in Stochastic Control Systems with Formal
                  Guarantees},
  booktitle    = {NeurIPS},
  year         = {2023}
}

@inproceedings{DBLP:conf/aaai/NadaliM0024,
  author       = {Alireza Nadali and
                  Vishnu Murali and
                  Ashutosh Trivedi and
                  Majid Zamani},
  title        = {Neural Closure Certificates},
  booktitle    = {{AAAI}},
  pages        = {21446--21453},
  publisher    = {{AAAI} Press},
  year         = {2024}
}

@inproceedings{DBLP:conf/cav/AbateGS24,
  author       = {Alessandro Abate and
                  Mirco Giacobbe and
                  Yannik Schnitzer},
  title        = {Bisimulation Learning},
  booktitle    = {{CAV} {(3)}},
  series       = {Lecture Notes in Computer Science},
  volume       = {14683},
  pages        = {161--183},
  publisher    = {Springer},
  year         = {2024}
}

@inproceedings{DBLP:conf/nips/AbateEG22,
  author       = {Alessandro Abate and
                  Alec Edwards and
                  Mirco Giacobbe},
  title        = {Neural Abstractions},
  booktitle    = {NeurIPS},
  year         = {2022}
}

@inproceedings{DBLP:conf/cav/AbateGR25,
  author       = {Alessandro Abate and
                  Mirco Giacobbe and
                  Diptarko Roy},
  title        = {Quantitative Supermartingale Certificates},
  booktitle    = {{CAV} {(2)}},
  series       = {Lecture Notes in Computer Science},
  volume       = {15932},
  pages        = {3--28},
  publisher    = {Springer},
  year         = {2025}
}

@inproceedings{DBLP:conf/aaai/NeustroevGL25,
  author       = {Grigory Neustroev and
                  Mirco Giacobbe and
                  Anna Lukina},
  title        = {Neural Continuous-Time Supermartingale Certificates},
  booktitle    = {{AAAI}},
  pages        = {27538--27546},
  publisher    = {{AAAI} Press},
  year         = {2025}
}

@inproceedings{DBLP:conf/nips/GiacobbeKPT25,
  author       = {Mirco Giacobbe and
                  Daniel Kroening and
                  Abhinandan Pal and
                  Michael Tautschnig},
  title        = {Let a Neural Network be Your Invariant},
  booktitle    = {NeurIPS},
  year         = {2025}
}

@inproceedings{DBLP:conf/cav/AbateGMS25,
  author       = {Alessandro Abate and
                  Mirco Giacobbe and
                  Christian Micheletti and
                  Yannik Schnitzer},
  title        = {Branching Bisimulation Learning},
  booktitle    = {{CAV} {(4)}},
  series       = {Lecture Notes in Computer Science},
  volume       = {15934},
  pages        = {161--184},
  publisher    = {Springer},
  year         = {2025}
}

@article{DBLP:journals/corr/abs-2505-15497,
  author       = {Frederik Baymler Mathiesen and
                  Nikolaus Vertovec and
                  Francesco Fabiano and
                  Luca Laurenti and
                  Alessandro Abate},
  title        = {Certified Neural Approximations of Nonlinear Dynamics},
  journal      = {CoRR},
  volume       = {abs/2505.15497},
  year         = {2025}
}

@inproceedings{DBLP:conf/atva/ChatterjeeGGKSSZ25,
  author       = {Krishnendu Chatterjee and
                  Amir Kafshdar Goharshady and
                  Ehsan Kafshdar Goharshady and
                  Mehrdad Karrabi and
                  Milad Saadat and
                  Maximilian Seeliger and
                  Dorde Zikelic},
  title        = {PolyQEnt: {A} Polynomial Quantified Entailment Solver},
  booktitle    = {{ATVA}},
  series       = {Lecture Notes in Computer Science},
  volume       = {16145},
  pages        = {411--424},
  publisher    = {Springer},
  year         = {2025}
}

\end{document}